\documentclass[preprint,authoryear,10pt,letterpaper]{elsarticle}
\usepackage[top=0.75in, bottom=0.75in, left=0.75in, right=0.75in]{geometry}
\usepackage[utf8]{inputenc}
\usepackage{amsmath,amsthm}
\usepackage{amssymb}

\usepackage{amsthm}
\usepackage{hyperref}
\usepackage{tabularx,booktabs}
\usepackage{multirow}
\usepackage{bigstrut}
\usepackage{lineno}
\usepackage{makecell}
\newtheorem{theorem}{Theorem}
\hypersetup{
    colorlinks=true,
    linkcolor=blue,
    urlcolor=blue,
    citecolor=blue,
}

\journal{ }
\makeatletter
\def\ps@pprintTitle{%
   \let\@oddhead\@empty
   \let\@evenhead\@empty
   \let\@oddfoot\@empty
   \let\@evenfoot\@oddfoot
}
\makeatother

\begin{document}
\begin{frontmatter}
\title{Real-time forecasting of metro origin-destination matrices with high-order weighted dynamic mode decomposition}

\author[label1]{Zhanhong Cheng}
\ead{zhanhong.cheng@mail.mcgill.ca}

\author[label2]{Martin Tr\'epanier}
\ead{mtrepanier@polymtl.ca}

\author[label1]{Lijun Sun\corref{cor1}}
\ead{lijun.sun@mcgill.ca}

\address[label1]{Department of Civil Engineering, McGill University, Montreal, QC H3A 0C3, Canada}
\address[label2]{Department of Mathematics and Industrial Engineering, Polytechnique Montreal, Montreal, QC H3T 1J4, Canada}

\cortext[cor1]{Corresponding author. Address: 492-817 Sherbrooke Street West, Macdonald Engineering Building, Montreal, Quebec H3A 0C3, Canada}

\begin{abstract}
    Forecasting the short-term ridership among origin-destination pairs (OD matrix) of a metro system is crucial in real-time metro operation. However, this problem is notoriously difficult due to the high-dimensional, sparse, noisy, and skewed nature of OD matrices. This paper proposes a High-order Weighted Dynamic Mode Decomposition (HW-DMD) model for short-term metro OD matrices forecasting. DMD uses Singular Value Decomposition (SVD) to extract low-rank approximation from OD data, and a low-rank high-order vector autoregression model is established for forecasting. To address a practical issue that metro OD matrices cannot be observed in real-time, we use the boarding demand to replace the unavailable OD matrices. Particularly, we consider the time-evolving feature of metro systems and improve the forecast by exponentially reducing the weights for old data. Moreover, we develop a tailored online update algorithm for HW-DMD to update the model coefficients daily without storing historical data or retraining. Experiments on data from a large-scale metro system show the proposed HW-DMD is robust to the noisy and sparse data and significantly outperforms baseline models in forecasting both OD matrices and boarding flow. The online update algorithm also shows consistent accuracy over a long time when maintaining an HW-DMD model at low costs.
\end{abstract}


\end{frontmatter}

\section{Introduction}
The metro is a green and efficient travel mode that plays an ever-important role in urban transportation. An accurate real-time ridership/demand forecast is crucial for a metro system's efficiency and reliability. With the wide application of smart card systems and all kinds of sensors, forecasting the real-time metro ridership has become a research hotspot in recent years. Existing research mainly focuses on forecasting the short-term (e.g., 15 or 30 minutes) boarding or alighting ridership at metro stations, such as \cite{wei2012forecasting, sun2015novel, li2017forecasting, chen2019subway, liu2019deeppf, zhang2020deep}. In contrast, forecasting the short-term ridership at origin-destination (OD) pairs of a metro system receives little attention. The ridership among all OD pairs of a metro system can be organized into a matrix. For simplicity, an ``OD matrix'' in this paper refers to the ridership of this matrix at a certain (short) time interval.

Forecasting metro OD matrices has much broader applications than the station-level ridership forecast. For example, by assigning OD matrices to a metro network, we can predict and thus regulate each metro train's crowdedness. The station-level boarding/alighting flow also can be calculated from an OD matrix. However, the real-time forecast of metro OD matrices is extremely difficult for the following reasons. (1) The first challenge is the \textit{high dimensionality}. The number of OD pairs of a metro system is squared of the number of stations, often tens of thousands in practice. Moreover, (2) short-term OD matrices of a metro system are \textit{sparse}, and the ridership/flow distribution within an OD matrix is highly \textit{skewed} (Fig.~\ref{fig:histogram}). Besides, unlike the boarding or alighting flow, (3) a metro system's OD matrices are not available in real-time (\textit{delayed data availability}). Because an OD matrix can only be obtained after all the trips belonging to the OD matrix have reached their destinations. Lastly, (4) the complex dynamics of a metro system are \textit{time-evolving}; a well-tuned model may have a short ``shelf life'' and has expensive retrain/re-tune costs in long-term maintenance. Although a few studies tried to forecast the real-time metro OD matrices by matrix factorization methods \citep{gong2018network, gong2020online} or deep learning models \citep{toque2016forecasting, zhang2021short, shen2020hybrid}, no existing solution overcomes all the four challenges above.

To address the above challenges, this paper proposes an effective model for real-time metro OD matrices forecasting. We use the Dynamic Mode Decomposition (DMD) \citep{schmid2010dynamic}, a recent advance in the fluid flow community, to extract the dominating dynamics from the high-dimensional noisy OD sequence. We extend the original DMD model by a high-order vector autoregression to incorporate long-term temporal correlations. In dealing with the delayed data availability problem, we replace the unavailable latest OD matrices with snapshots of boarding flow. We consider the time-evolving dynamics and introduce a forgetting ratio to reduce the weights of old data exponentially. We name the proposed model High-order Weighted Dynamic Mode Decomposition (HW-DMD). Moreover, we develop a tailored online update algorithm that updates an HW-DMD's coefficients daily without storing historical data or retraining the model, which greatly reduces the model maintenance costs for long-term implementations. Finally, the proposed model is tested on a Guangzhou metro smart card dataset with 159 stations. Experiments show HW-DMD excellently handles the sparse, skewed, and noise OD data and significantly outperforms baseline models in both the OD matrices and the boarding flow forecast. The online update algorithm also shows consistent accuracy in updating an HW-DMD model over a long period.

The online HW-DMD model is applied to the metro OD matrices forecast problem, but it can be readily applied to general (high-dimensional) traffic flow forecast problems, such as in recent studies about DMD-based traffic flow forecasting \citep{avila2020data, yu2020low}. We summary main contributions of this paper as follows:
\begin{itemize}
    \item This paper proposes an HW-DMD model that addresses various difficulties of the real-time metro OD matrices forecasting. Experiments show the forecast of HW-DMD is significantly better than conventional models.
    \item The time-evolving dynamics of a transportation system and the maintenance of a forecast model are often ignored in the literature. This paper considers the time-evolving feature of a metro system by reducing the weights for old data and shows improved performance. An online update algorithm is proposed to reduce the long-term maintenance cost of the HW-DMD model under a time-evolving metro system.
\end{itemize}

The remainder of this paper is organized as follows. We review related work on the short-term OD matrices forecasting in section \ref{cha5_sec:literature}. Next, a description of the metro OD matrices forecasting problem is presented in section \ref{sec:problem}. Section \ref{sec:DMD} briefly introduces the DMD algorithm, which serves as the base for the proposed HW-DMD model. Section \ref{sec:HW-DMD} is the core part of this paper, where the model specification, estimation, and the online update method for HW-DMD are elaborated. Section \ref{sec:experiments} is about experiments on the Guangzhou metro dataset. Conclusions and future directions are summarized in section \ref{sec:conclusions}.

\section{Related Work}\label{cha5_sec:literature}

In the literature, only a few studies have explored the real-time OD matrix forecasting problem for a ``metro'' system. Therefore, we extend the range to OD demand forecasting for general road transportation modes, such as the ride-hailing system and the highway tolling system. Note that for a ride-hailing system, the origins and destinations are often defined as zones on a grid.

Matrix/tensor factorization is an effective method to tackle the high-dimensionality problem of OD matrix forecasting. For example, \cite{ren2017efficient} applied Canonical Polyadic (CP) decomposition to an $origin \times destination \times vehicle\_type \times time$ tensor from highway tolling data. Time series models were then built on the latent temporal matrix to forecast OD matrices. \cite{dai2018short} and \cite{liu2020dynamic} used principal component analysis (PCA) to reduce the dimensionality of OD data and applied several machine learning models to the reduced data for OD flow forecasting. \cite{gong2020online} developed a matrix factorization model to forecast the OD matrices of a metro system. Their work highlights a solution to the delayed data availability problem and various spatial and temporal regularization techniques are introduced to improve the forecast. In summary, the matrix/tensor factorization-based OD matrix forecasting consists of two components: (1) a dimensionality reduction by factorization and (2) a forecasting model applied to the reduced data. 

Deep learning is another mainstream method for OD matrix forecasting. In an early study, \cite{toque2016forecasting} used Long Short-Term Memory (LSTM) networks to forecast the OD matrices of a transit network. They only applied the model to selected high-flow OD pairs because of the high dimensionality and sparsity problems. Convolutional Neural Networks (CNN) and Graph Convolutional Networks (GNN) are two deep learning models that greatly reduce the model size compared with a fully connected neural network. Recently, using CNN/GCN to capture spatial correlations and LSTM to capture temporal correlations started to become the ``standard configuration'' for deep learning-based OD matrix forecasting. For example, \cite{chu2019deep} used multi-scale convolutional LSTM to forecast the real-time taxi OD demand, and \cite{wang2019origin, wang2020multi} used multi-task learning to improve the OD flow forecast of GCN+LSTM networks. A large body of literature focused on better utilizing the spatial/semantic correlations by optimizing the GNN structure or incorporating side information. Such as the local spatial context used by \cite{liu2019contextualized}, the Spatio-Temporal Encoder-Decoder Residual Multi-Graph Convolutional network (ST-ED-RMGC) proposed by \cite{ke2021predicting}, and the Dynamic Node-Edge Attention Network (DNEAT) developed by \cite{zhang2021dneat}. Some studies combined deep learning models with other models to complement each other. In this direction, \citep{xiong2020dynamic} combined GCN with Kalman filter to forecast the OD matrices of a Turnpike network. \cite{shen2020hybrid} mixed CNN with a Gravity model to forecast OD matrices of a metro system. \cite{hu2020stochastic} considered the travel time between OD pairs as a stochastic variable, and developed a stochastic OD matrix forecasting model based on tensors factorization and GCN. \cite{noursalehi2021dynamic} used discrete wavelet transform to decompose OD matrices into frequency domain; the outputs were fed into CNN and Convolutional-LSTM networks for forecasting.

The performances of deep learning models are often impaired by the noise in sparse metro OD matrices. To reduce the impact of the noise, \cite{zhang2019short, zhang2021short} developed a metric called OD attraction degree (ODAD) to mask insignificant OD pairs. \cite{zhang2019short} showed that masking near-zero OD pairs improves the forecasting accuracy of an LSTM. Based on ODAD, \cite{zhang2021short} developed a Channel-wise Attentive Split-CNN (CAS-CNN) model for metro OD matrix forecasting. Another merit of this work is they considered the delayed data availability problem.

In summary, matrix/tensor factorization, CNN, and GCN all aim to reduce model size while maintaining spatial/temporal correlations/dependencies. The HW-DMD model proposed in this paper belongs to the matrix factorization category. Although some ride-hailing systems may not have the delayed data availability problem, most research essentially omitted this problem for simplicity. Particularly, RNN-based deep learning models can barely work without the most recent OD matrices as inputs. In dealing with the delayed data availability problem, existing solutions \citep{gong2020online, zhang2021short, xiong2020dynamic} used alternative quantities (e.g., boarding ridership, link flow) to compensate for the unavailable OD information. We also adopt this approach in the proposed model.


\section{Problem Description}\label{sec:problem}
Many modern metro systems record passengers' entry and exit information using smart cards. We thus know the origin and destination stations, the start and end time for every trip in such a system. Given a fixed time interval (30 minutes in this paper), we denote by $o_{t,i,j}$ the number of trips that depart from station $i$ at the $t$-th interval to station $j$. We call $o_{t,i,j}$ an \textit{OD flow}. Next, we can describe the number of trips between every OD pair in the system at the $t$-th time interval by an \textit{OD matrix}
\begin{equation*}
    O_t = \left[ \begin{matrix}
        o_{t,1,1}&		\cdots&		o_{t,1,s}\\
        \vdots&		\ddots&		\vdots\\
        o_{t,s,1}&		\cdots&		o_{t,s,s}\\
    \end{matrix} \right]\in \mathbb{R}^{s\times s},
\end{equation*}
where $s$ is the number of metro stations. The diagonal elements of a metro OD matrices are always zero. We keep these zero elements because they have a negligible effect on the forecast. In our model, OD matrices are organized in a vector form
\begin{equation*}
    \mathbf{f}_t = \operatorname{vec}(O_t) = [o_{t,1,1}, \cdots, o_{t,s,1}, o_{t,1,2},\cdots, o_{t,s,2}, \cdots, o_{t,1,s}, \cdots, o_{t,s,s}]^\top\in \mathbb{R}^{n},
\end{equation*}
where $n=s \times s$ is the number of OD pairs. For convenience, we refer to $\mathbf{f}_t$ as an \textit{OD snapshot}.

Note that OD snapshots are aggregated by the time when passengers enter the system; the exit time might be in a different time interval. Therefore, the true OD snapshot for interval $t$ can only be obtained after all those passengers entered at interval $t$ have reached their destinations; it cannot be observed in real-time (i.e., the delayed data availability). In other words, we often do not have access to $\mathbf{f}_{t}$ when forecasting $\mathbf{f}_{t+1}$. In contrast, the boarding (entering) flow--another important quantity--is observable in real-time. We denote by $b_{t,i}$ the number of passengers entering station $i$ at interval $t$. In fact, we have $b_{t,i}=\sum_{j}{o_{t,i,j}}$. We define a \textit{boarding snapshot} as a vector $\mathbf{b}_{t}=[b_{t,1}, b_{t,2}, \cdots, b_{t,s}]^\top$.

The OD matrices/flow forecasting problem is to forecast future OD snapshots $\mathbf{f}_{t+1}, \mathbf{f}_{t+2},\allowbreak \cdots, \mathbf{f}_{t+L}$ given a sequence of available historical OD snapshots $\mathbf{f}_{1}, \mathbf{f}_{2}, \cdots, \mathbf{f}_{t}$ and boarding snapshots $\mathbf{b}_{1}, \mathbf{b}_{2}, \cdots, \mathbf{b}_{t}$. The reason for using boarding snapshots is to compensate for the delayed data availability problem of recent OD snapshots.

\section{Dynamic Mode Decomposition}\label{sec:DMD}
Dynamic Mode Decomposition \citep[DMD,][]{schmid2010dynamic} is developed by the fluid dynamic community to extract dynamic features from high-dimensional data. To better illustrate our forecasting model, we briefly introduce DMD in this section.

Consider using a linear dynamical system $\mathbf{f}_{i} \approx A \mathbf{f}_{i-1}$ for OD flow forecasting. Similar to many fluid problems, $n$ is huge for an OD snapshot and even storing $A\in \mathbb{R}^{n\times n}$ can be prohibitive. Therefore, DMD outputs the (leading) eigenvalues and eigenvectors of $A$ without calculating the expensive $A$. The eigenvectors of $A$ are referred to as the DMD modes and have clear physical meaning. Each DMD mode is associated with an oscillation frequency and a decay/growth rate determined by its eigenvalue. DMD is also connected to Koopman theory and can model complex non-linear systems by constructing proper measurements \citep{rowley2009spectral}. There are many variant algorithms for DMD. We only present the \textit{exact DMD} proposed by \cite{tu2014dynamic}, which is closely related to this paper.

We arrange OD snapshots into $m$-column matrices $Y_i = [\mathbf{f}_{i-m+1}, \mathbf{f}_{i-m+2}, \cdots , \mathbf{f}_{i}] \in \mathbb{R}^{n \times m}$. Typically, $m \ll n$. The linear dynamical system follows $Y_{t} \approx A Y_{t-1}$. The exact DMD seeks the leading eigenvalues and eigenvectors of the best-fit linear operator $A$ by the following procedure.
\begin{itemize}
    \item[1.] Compute the truncated singular value decomposition (SVD) of $Y_{t-1}\approx U\Sigma V^{\top}$, where $U\in \mathbb{R}^{n\times r}$, $\Sigma\in \mathbb{R}^{r\times r}$, and $V\in \mathbb{R}^{m\times r}$ and $r\ll m$.
    \item[2. ] Instead of computing the full matrix $A = Y_{t} Y_{t-1}^{+} \approx Y_{t}V\Sigma^{-1}U^{\top}$.\footnote{$(\cdot)^{+}$ denotes the Moore-Penrose inverse of a matrix.} We define a reduced matrix $\tilde{A}=U^{\top}AU \approx U^{\top}Y_{t}V\Sigma^{-1} \in \mathbb{R}^{r \times r}$. It can be proved that $\tilde{A}$ and $A$ have the same nonzero leading eigenvalues \citep{tu2014dynamic}.
    \item[3. ] Compute the eigenvalue decomposition $\tilde{A}W=W\Lambda$. The entries of the diagonal matrix $\Lambda$ are also the eigenvalues of the full matrix $A$.
    \item[4. ] The DMD modes (eigenvectors of $A$) can be obtained by $\Phi=Y_t V \Sigma^{-1} W$.
\end{itemize}

\begin{figure}
\begin{center}
\includegraphics[]{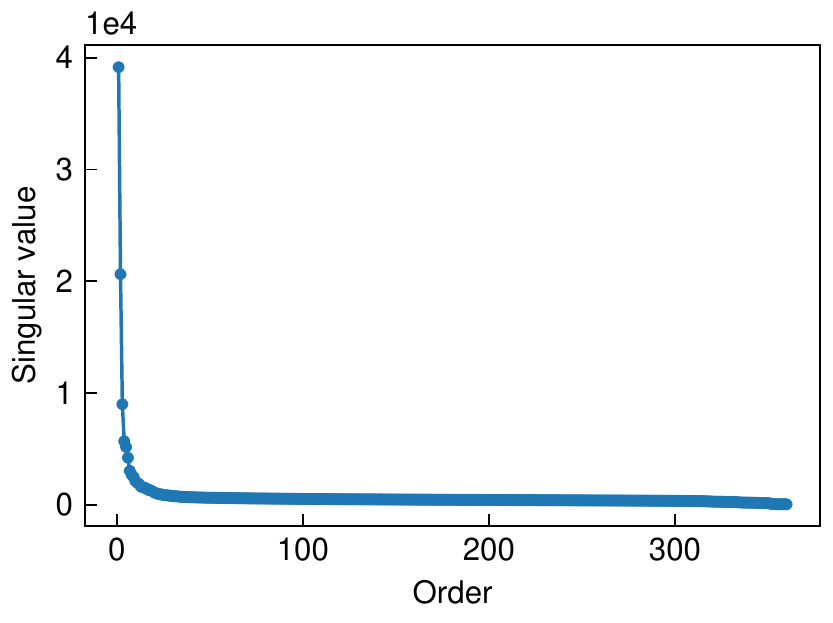}
\caption{The singular values of a ten-day-length $Y_{t-1}$ collected from Guangzhou metro smart card system.}\label{fig:singular value}
\end{center}
\end{figure}

Figure~\ref{fig:singular value} shows the singular values of a ten-day-length $Y_{t-1}$ from the Guangzhou metro system. A few leading singular values explain a significant portion of the variance, confirming the low-rank feature of OD snapshots data. DMD-based model can thus greatly reduce the dimensionality/complexity of such a dynamic system. However, the exact DMD has some limitations for the OD flow forecasting problem. Firstly, the complex temporal correlation of OD flow cannot be well captured by a linear dynamical system. Moreover, using the last OD snapshot is impractical since OD snapshots cannot be observed in real-time. To address these problems, we propose our solution in the next section.

\section{High-order Weighted Dynamic Mode Decomposition}\label{sec:HW-DMD}

\subsection{Model specification}\label{sec:specification}
The forecasting formula of an exact DMD amounts to a high-dimensional vector autoregression of order 1. However, the latest OD snapshots are unknown at the time of forecasting. Therefore, we use the two latest snapshots of the boarding flow as a replacement. We regard OD snapshots of three or more intervals ago as available; because we find more than 96\% trips in our data set are completed within one hour (two lags). And we can use a high-order vector autoregression to capture the long-term correlations in OD snapshots. The forecasting model follows
\begin{equation}
    \mathbf{f}_{i} \approx A_{t, 1}\mathbf{f}_{i - q_1} + A_{t, 2}\mathbf{f}_{i-q_2} + \cdots + A_{t, h}\mathbf{f}_{i-q_h} + A_{t, b1}\mathbf{b}_{i-1} + A_{t, b2}\mathbf{b}_{i-2} \quad \forall i \in \{q_h+1, q_h+2, \cdots, t \}, \label{eq:model}
\end{equation}
where time lags for OD snapshots are positive integers satisfying $3\leq q_1 < \cdots < q_h < t$. Note that coefficient matrices $A_{t,1}, \cdots, A_{t,h}\in \mathbb{R}^{n\times n}$ and $A_{t,b1}, A_{t,b2}\in\mathbb{R}^{n\times s}$ are estimated using the data up to the latest ($t$-th) time interval; they are re-estimated when new data become available. This allows model coefficients to be time-varying. We will introduce how to update coefficient matrices using new observations without storing historical data in Section~\ref{sec:online}.

To express Eq.~\eqref{eq:model} in a concise matrix form, let $Y_i = [\mathbf{f}_{i-m+1}, \mathbf{f}_{i-m+2}, \cdots , \mathbf{f}_{i}]$ and $B_{i}=[\mathbf{b}_{i-m+1}, \mathbf{b}_{i-m+2}, \cdots , \mathbf{b}_{i}]$, where $m=t-q_h$ is the number of target snapshots. Then, Eq.~\eqref{eq:model} is equivalent to
\begin{align}
    Y_t & \approx A_{t, 1}Y_{t - q_1} + A_{t, 2}Y_{t-q_2} + \cdots + A_{t, h}Y_{t-q_h} + A_{t, b1}B_{t-1}+ A_{t, b2}B_{t-2}\label{eq:var matrix}\\
    & =[A_{t,1},\  A_{t,2}, \cdots A_{t,h},\ A_{t, b1},\ A_{t, b2}]
    \left[ \begin{array}{c}
        Y_{t-q_1}\\
        Y_{t-q_2}\\
        \vdots \\
        Y_{t-q_h}\\
        B_{t-1}\\
        B_{t-2}
    \end{array} \right]\\
    & = G_t X_t,
\end{align}
where $G_t \in \mathbb{R}^{n\times (hn+2s)}$ and $X_t \in \mathbb{R}^{(hn+2s)\times m}$ are augmented matrices for coefficients and data, respectively. Note that with this approach, we model forecasting as a regression problem without considering the inter-sequence dependence.

We next introduce a forgetting ratio $\rho$ ($0<\rho\le 1$) that assigns small weights on snapshots to past days. This is because the dynamics of the system may change over time and we prefer to use the most recent dynamics to achieve accurate forecasting. The matrix $G_t$ can be solved by the following optimization problem
\begin{equation}
    \min_{G_t} \sum_{i=1}^{m}\rho^{\mathrm{day}(m)-\mathrm{day}(i)} {\left\| \mathbf{y}_i - G_t \mathbf{x}_i \right\|_F^2}, \label{eq:weighted loss}
\end{equation}
where $\mathbf{y}_i$ and $\mathbf{x}_i$ are the $i$-th column of $Y_t$ and $X_t$, respectively; $\mathrm{day}(i)$ represents the day of the snapshot $\mathbf{y}_i$. We assign the same weight for snapshots of the same day. The weight $\rho^{\mathrm{day}(m)-\mathrm{day}(m)}=\rho^{0}=1$ for the latest OD snapshot. For a snapshot in $j$ days ago, the weight is $\rho^{j}$, which decreases exponentially. This weighting idea is similar to works by \citet{alfatlawi2020incremental}, \citet{zhang2019online}, and \citet{kwak2020travel}. For convenience, we define $\sigma = \sqrt{\rho}$ and the weighted version of $Y_t$ and $X_t$ as
\begin{align*}
    Y_t^w &= [\sigma^{\mathrm{day}(m)-\mathrm{day}(1)} \mathbf{y}_1, \ \sigma^{\mathrm{day}(m)- \mathrm{day}(2)} \mathbf{y}_2, \cdots,  \mathbf{y}_m], \\
    X_t^w &= [\sigma^{\mathrm{day}(m)-\mathrm{day}(1)} \mathbf{x}_1, \ \sigma^{\mathrm{day}(m)-\mathrm{day}(2)} \mathbf{x}_2, \cdots,  \mathbf{x}_m].
\end{align*}
Then, the optimization problem in Eq.~\eqref{eq:weighted loss} becomes an ordinary least squares problem
\begin{equation}
    \min_{G_t} \left\| Y_t^w - G_t X_t^w\right\|_F^2. \label{eq:weighted loss2}
\end{equation}

\begin{figure}
\begin{center}
\includegraphics[width=0.9\textwidth]{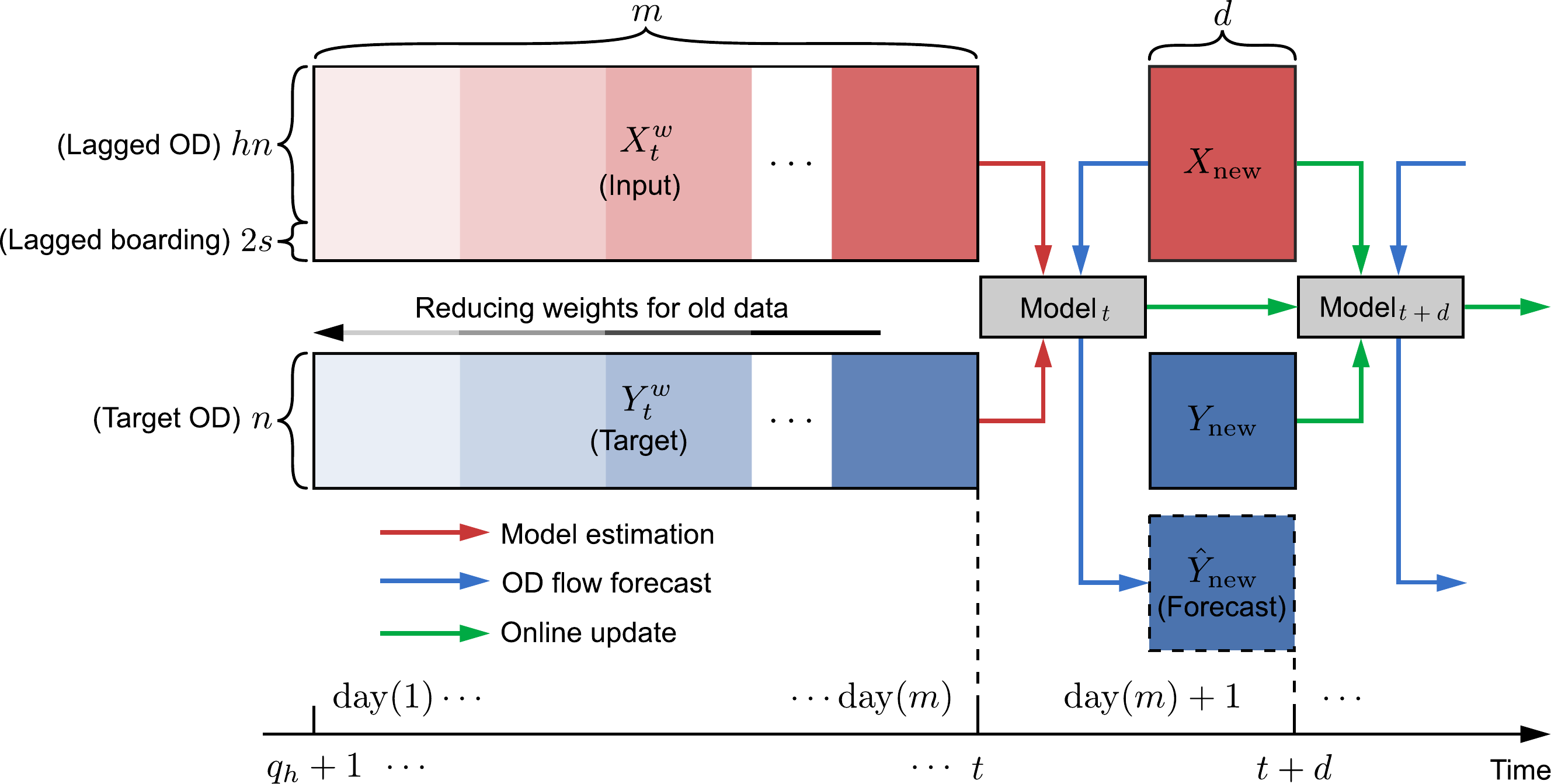}
\caption{Model framework for HW-DMD. Model input $X$ contains $hn$ rows for lagged OD snapshots and $2s$ rows for lagged boarding snapshots. Columns in $Y$ and $\hat{Y}$ are, respectively, real and forecasted snapshots for OD flow. Model coefficients are estimated by weighted historical data ($X_t^w$ and $Y_t^w$) and updated daily whenever new data come.}\label{fig:framework}
\end{center}
\end{figure}

Figure~\ref{fig:framework} summarizes the overall structure of the proposed higher-order weighted DMD (HW-DMD) framework. The underlying forecasting model is a high-order vector autoregression with the boarding flow as extra inputs. A forgetting ratio is introduced to decrease the weights of past data exponentially on a daily basis. In Section~\ref{sec:estimation}, we will introduce a dimensionality reduction technique based on DMD to find a low-rank solution for this large model (w.r.t. number of parameters). Instead of full matrices $A_{t, (\cdot)}$, we seek $\tilde{A}_{t, (\cdot)}$---much smaller matrices---to capture the system's dynamic. Finally, an online update method is proposed in Section~\ref{sec:online} to update the model coefficients incrementally without storing historical data. This provides a memory-saving solution that maintains an up-to-date model.
Note that the same model framework can be easily extended to incorporate higher-order boarding flow or other external covariates (e.g., days of the week, alighting flow, holidays). For example, we can represent days of the week by one-hot encoding $\mathbf{w}_i\in \mathbb{R}^{7\times 1}$ and add an additional regression term $A_{t,w} \mathbf{w}_i$ to Eq.\eqref{eq:model} to incorporate the weekly pattern. This paper only presents the model specified in Eq.~\eqref{eq:model} for illustration.

\subsection{Model estimation}\label{sec:estimation}
We prefer a low-rank approximation of $G_t$ over a full matrix of the optimal solution of Eq.~\eqref{eq:weighted loss2}. This is because storing the large full matrix is prohibitive, and the optimal solution often leads to overfitting problems, especially for the sparse and noisy OD data. Luckily, we can find a pretty good approximation thanks to the inherent low-rank nature of OD data.

Similar to the exact DMD, we first compute the truncated SVD on the weighted augmented data matrix $X^w_t \approx U_{X} \Sigma_{X} V_{X}^{\top}$, where we keep the $r_{X}$ ($r_{X} \ll m$) largest singular values and $U_{X}\in \mathbb{R}^{(hn+2s) \times r_{X}}$, $\Sigma_{X}\in \mathbb{R}^{r_{X} \times r_{X}}$, $V_{X}\in \mathbb{R}^{m \times r_{X}}$. As shown in Figure~\ref{fig:singular value}, a few leading singular values can well capture the entire data. Therefore, an approximation for coefficient matrices is
\begin{equation}
    G_t = Y_{t}^{w}X_{t}^{w+} \approx Y^{w}_{t} V_{X} \Sigma_{X}^{-1} U_{X}^{\top}, \label{eq:Gt}
\end{equation}
\begin{multline}
    [A_{t,1}, \cdots \ A_{t,h}, \ A_{t, b1}, \ A_{t, b2}] \approx  \\
    [Y^{w}_{t} V_{X} \Sigma_{X}^{-1} U_{X,1}^{\top}, \cdots, Y^{w}_{t} V_{X} \Sigma_{X}^{-1} U_{X,h}^{\top},\ Y^{w}_{t} V_{X} \Sigma_{X}^{-1} U_{X,b1}^{\top},\ Y^{w}_{t} V_{X} \Sigma_{X}^{-1} U_{X,b2}^{\top}], \label{eq:A}
\end{multline}
where $U_{X}^{\top} = [U_{X,1}^{\top}, \cdots, U_{X,h}^{\top}, \ U_{X,b1}^{\top}, \ U_{X,b2}^{\top}]$, $U_{X,1}, \cdots, U_{X,h} \in \mathbb{R}^{n \times r_X}$, and $U_{X,b1}, U_{X,b2} \in \mathbb{R}^{s \times r_X}$. This step uses the result from a truncated SVD to replace the original $X_t^w$, which reduces the impact of the noise in the data.

The results computed from Eq.~\eqref{eq:Gt} and Eq.~\eqref{eq:A} are still prohibitive. Therefore, for each column in $Y_{t}^{w}$, we seek a transformation $\mathbf{y}_i^w \rightarrow \tilde{\mathbf{y}}_i^w$ such that $\tilde{\mathbf{y}}^w_{i} \in \mathbb{R}^{r_{Y}}$ with $r_{Y} \ll n$. In doing so, we compute another rank-$r_Y$ truncated SVD of the target matrix $Y^{w}_{t} \approx U_{Y} \Sigma_{Y} V_{Y}^{\top}$. The columns of $U_Y$ form an orthonormal basis; thus, the transformation $\tilde{\mathbf{y}}^w_i = U_{Y}^{\top} \mathbf{y}^w_i$ compute the coordinates of $\mathbf{y}^w_i$ on this basis, which compresses $y_i^w$ from $\mathbb{R}^n$ to $\mathbb{R}^{r_Y}$. We can project coefficient matrices onto the same basis $U_Y$ to greatly reduce the dimensionality:
\begin{align}
    \tilde{A}_{t,i} &= U_{Y}^{\top} A_{t,i} U_{Y} \approx  U_{Y}^{\top} Y^{w}_{t} V_{X} \Sigma_{X}^{-1} U_{X,i}^{\top} U_{Y}, \quad \forall i \in \{1,2,\cdots, h\}, \label{eq:coeff_tilde} \\
    \tilde{A}_{t,bj} &= U_{Y}^{\top} A_{t,bj} \approx  U_{Y}^{\top} Y^{w}_{t} V_{X} \Sigma_{X}^{-1} U_{X,bj}^{\top}, \quad \forall j \in \{1,2\}, \label{eq:coeff_tilde_b}
\end{align}
where $\tilde{A}_{t,i} \in \mathbb{R}^{r_Y\times r_Y}$ and $\tilde{A}_{t,bj}\in \mathbb{R}^{r_Y \times s}$. Finally, we can write the model of Eq.~\eqref{eq:var matrix} in the reduced-order subspace
\begin{equation*}
    \tilde{Y}_{t} \approx \tilde{A}_{t,1}\tilde{Y}_{t - q_1} + \tilde{A}_{t,2}\tilde{Y}_{t-q_2} + \cdots + \tilde{A}_{t,h}\tilde{Y}_{t-q_h} + \tilde{A}_{t,b1}B_{t-1} + \tilde{A}_{t,b1}B_{t-2},
\end{equation*}
where $\tilde{Y}_i = U_Y^{\top} Y_{i}$. The final forecast of an OD snapshot $\hat{\mathbf{y}}_i$ can be calculated by transforming back to the original basis by $\hat{\mathbf{y}}_i = U_Y \tilde{\mathbf{y}}_i$. With the reduced coefficient matrices $\tilde{A}_{t,(\cdot)}$ and projection bases $U_Y$, we avoid calculating and storing the giant coefficient matrices $A_{t,(\cdot)}$.

DMD-based estimation is different from common dimensionality reduction techniques in several ways. For many matrix-factorization-based models and dynamic factor models, a forecasting model is estimated after performing dimensionality reduction \citep[e.g.,][]{ren2017efficient}, or latent factors are constructed by keeping the most temporal dynamics \citep[e.g.,][]{forni2000generalized, lam2011estimation, yu2016temporal}; the forecast ability is designed on the latent (size-reduced) data for these models. In contrast, DMD-based methods first estimate a forecasting model by a least-square fit of rank-reduced full-size data (i.e., Eq.~\eqref{eq:Gt}), next reduce the dimensionality of the linear operator by projecting to leading SVD modes (i.e., Eq.~\eqref{eq:coeff_tilde}--\eqref{eq:coeff_tilde_b}); the resulting linear operator captures the dynamics of the rank-reduced full-sized data. Although the forecast value $\hat{\mathbf{y}}_i$ by an HW-DMD is restricted on the column space of $U_Y$, it is already the best approximation in $\mathbb{R}^{r_Y}$ (in terms of Frobenius norm \citep{eckart1936approximation}) because the basis is determined by leading singular vectors. Besides, the rank truncation for the data also eases the noise and the overfitting problem. As noted by \citet{schmid2010dynamic}, accurate identification of more than the first couple modes can be difficult on noisy data sets without this truncation step.

The major computational cost in parameter estimation of HW-DMD is the SVD part. Current numerical software can solve large-scale SVD very efficiently. Therefore, estimating the HW-DMD model is very fast. We can further derive the eigenvalues and eigenvectors of coefficient matrices $A_{t,i}$ \citep{proctor2016dynamic}. But this step is not necessary for our task, since they are not used to generate the forecast and there is no clear physical meaning for eigenvectors in a high-order vector autoregression.

\subsection{Online update}\label{sec:online}
A model trained by dated data may not reflect the recent dynamic in a system. Instead of retraining using entire data, we develop an online algorithm that updates HW-DMD day by day with new observations without storing historical data, as shown in Figure~\ref{fig:framework}. Similar algorithms for online DMD have been developed by \citet{hemati2014dynamic}, \citet{zhang2019online}, and \citet{alfatlawi2020incremental}. We extend the online DMD update algorithm to a high-order weighted version.

To illustrate the update algorithm, we need to reorganize Eq.~\eqref{eq:Gt}--\eqref{eq:coeff_tilde_b}. Let $\tilde{X}_i^w=U_X^{\top}X_i^w$ and $\tilde{Y}_i^w=U_Y^{\top}Y_i^w$ be the projection of data to the coordinates of $U_X$ and $U_Y$, respectively. Using the fact $(U_X\tilde{X}_t^{w})^{+}=V_{X} \Sigma_{X}^{-1} U_{X}^{\top}$, we can rewrite Eq.~\eqref{eq:Gt} as
\begin{align*}
    G_t  & \approx Y^{w}_{t} (U_X\tilde{X}_t^{w})^{+} \\
               & = Y^w_t \tilde{X}_t^{w\top} \left(\tilde{X}_t^w \tilde{X}_t^{w\top}\right)^{+} U_{X}^{\top}.
\end{align*}
Therefore, Eq.~\eqref{eq:coeff_tilde} and \eqref{eq:coeff_tilde_b} becomes
\begin{align}
    \tilde{A}_{t,i} &\approx \tilde{Y}^w_t \tilde{X}_t^{w\top} \left(\tilde{X}^w_t \tilde{X}_t^{w\top}\right)^{+} U_{X,i}^{\top} U_{Y} = P Q_{X}^{+} U_{X,i}^{\top} U_{Y}\quad \forall i \in \{1, \cdots,h\},\label{eq:stream Gt} \\
    \tilde{A}_{t,bj} &\approx \tilde{Y}^w_t \tilde{X}_t^{w\top} \left(\tilde{X}^w_t \tilde{X}_t^{w\top}\right)^{+} U_{X,bj}^{\top}= P Q_{X}^{+} U_{X,bj}^{\top}\quad \forall j \in \{1, 2\},
\end{align}
where $P=\tilde{Y}^w_t \tilde{X}_t^{w\top} \in \mathbb{R}^{r_Y \times r_X}$ and $Q_{X} = \tilde{X}^w_t \tilde{X}_t^{w\top} \in \mathbb{R}^{r_X \times r_X}$.

To facilitate the online update, we define an additional matrix $Q_{Y} = \tilde{Y}^w_t \tilde{Y}_t^{w\top}\in \mathbb{R}^{r_Y \times r_Y}$. After the reorganization, model coefficients are represented by three ``core'' matrices $P$, $Q_X$, $Q_Y$ and two projection matrices $U_X$, $U_Y$. Note these matrices are also time-varying. For simplicity, we omit the $t$ subscript and regard they are always ``up-to-date''. Moreover, there are two important properties for the core matrices.
\begin{theorem}\label{p1}
    Given new observations $Y_{\mathrm{new}}\in \mathbb{R}^{n \times d}$ and $X_{\mathrm{new}}\in \mathbb{R}^{(hn+2s)\times d}$ from a new day, where $d$ is the number of snapshots per day. Under the same projection matrices, the new core matrices can be updated by
    \begin{align}
        P &\leftarrow \rho P + \tilde{Y}_{\mathrm{new}}\tilde{X}_{\mathrm{new}}^{\top},\label{eq:update P}\\
        Q_X &\leftarrow \rho Q_X + \tilde{X}_{\mathrm{new}}\tilde{X}_{\mathrm{new}}^{\top},\label{eq:update Q_X}\\
        Q_Y &\leftarrow \rho Q_Y + \tilde{Y}_{\mathrm{new}}\tilde{Y}_{\mathrm{new}}^{\top},\label{eq:update Q_Y}
    \end{align}
    where $\tilde{X}_{\mathrm{new}}=U_X^{\top}X_{\mathrm{new}}$ and $\tilde{Y}_{\mathrm{new}}=U_Y^{\top}Y_{\mathrm{new}}$.
\end{theorem}

\begin{proof}
    Given new observations $Y_{\mathrm{new}}\in \mathbb{R}^{n \times d}$ and $X_{\mathrm{new}}\in \mathbb{R}^{(hn+2s)\times d}$ from a new day, Under the same projection matrices, the new core matrix $P$ can be computed by
    \begin{align*}
        \tilde{Y}^w_{t+d} \tilde{X}_{t+d}^{w \top}
        &= [\sigma \tilde{Y}^w_{t},\ U_Y^{\top}Y_{\mathrm{new}}][\sigma \tilde{X}^w_{t},\ U_X^{\top} X_{\mathrm{new}}]^{\top}\\
        &= [\sigma \tilde{Y}^w_{t},\ \tilde{Y}_{\mathrm{new}}][\sigma \tilde{X}^w_{t},\ \tilde{X}_{\mathrm{new}}]^{\top}\\
         &=\sigma^2 \tilde{Y}^w_{t} \tilde{X}^{w \top}_{t} + \tilde{Y}_{\mathrm{new}} \tilde{X}_{\mathrm{new}}^{\top}\\
         & = \rho P + \tilde{Y}_{\mathrm{new}} \tilde{X}_{\mathrm{new}}^{\top}.
    \end{align*}
Therefore, $P$ can be updated by $P\leftarrow \rho P + \tilde{Y}_{\mathrm{new}} \tilde{X}_{\mathrm{new}}^{\top}$. Similar proof applies to $Q_X$ and $Q_Y$.
\end{proof}

\begin{theorem}\label{p2} Denote by $\bar{Y}_t^w = U_Y\tilde{Y}_{t}^{w}$ the recovered data from the reduced data. If $\mathbf{v}_i$ is the $i$-th eigenvector of $Q_Y$, then $U_Y \mathbf{v}_i$ is the $i$-th left singular vector of  $\bar{Y}_t^w$. The same property applies to $Q_X$ and $\bar{X}_t^w = U_X\tilde{X}_{t}^{w}$.
\end{theorem}

\begin{proof}
Compute SVD $\bar{Y}_t^w = \bar{U} \bar{\Sigma} \bar{V}^{\top}$, then
\begin{align}
    & \bar{Y}_t^w \bar{Y}_t^{w\top} = \bar{U} \bar{\Sigma} \bar{V}^{\top} \bar{V} \bar{\Sigma}^{\top} \bar{U}^{\top} = \bar{U} \left(\bar{\Sigma}\bar{\Sigma}\right) \bar{U}^{\top},\\
    & \left(\bar{Y}_t^w \bar{Y}_t^{w\top}\right) \bar{U} =\bar{U} \left(\bar{\Sigma}\bar{\Sigma}\right) = \bar{U} \bar{\Lambda}\label{eq:appendix1}.
\end{align}
Therefore, columns of $\bar{U}$ are the eigenvectors of $\bar{Y}_t^w \bar{Y}_t^{w\top}$ and the left singular vectors of $\bar{Y}_t^w$. Substitute $\bar{Y}_t^w \bar{Y}_t^{w\top}=U_Y Q_Y U_Y^{\top}$ to Eq.~\eqref{eq:appendix1}, we have
\begin{align*}
    & \left(U_Y Q_Y U_Y^{\top} \right)\bar{U}= \bar{U} \bar{\Lambda}, \\
    & Q_Y \left(U_Y^{\top} \bar{U}\right)= \left(U_Y^{\top} \bar{U}\right) \bar{\Lambda}.
\end{align*}
Define $V=U_Y^{\top} \bar{U}$. Then, each column $\mathbf{v}_i$ in $V$ is a eigenvector for $Q_Y$ and $U_Y \mathbf{v}_i=U_Y \left(U_Y^{\top} \bar{\mathbf{u}}_i\right) = {\mathbf{u}}_i$ is a singular vector of $\bar{Y}_t^w$.
\end{proof}

Theorem~\ref{p1} is used to update the core matrices in a memory-saving way. Theorem~\ref{p2} indicates we can use the eigenvectors of $Q_Y$ to approximate the left singular vectors of $Y_t^w$ (because $Y_t^w \approx \bar{Y}_t^w$), which is crucial for updating the projection matrices. Based on these properties, we summarize the online update algorithm in the following three steps.
\begin{itemize}
    \item[1.] \textbf{Expand projection matrices}. Let $E_Y = Y_{\mathrm{new}}-U_{Y}U_{Y}^{\top}Y_{\mathrm{new}}$ and $E_X = X_{\mathrm{new}} - U_{X}U_{X}^{\top}X_{\mathrm{new}}$ be the residuals that cannot be represented by the column space of $U_X$ and $U_Y$. To incorporate these residuals, we expand projection matrices by $U_X\leftarrow[U_X, U_{E_X}]$ and $U_Y\leftarrow[U_Y, U_{E_Y}]$, where $U_{E_{X}}$ and $U_{E_{Y}}$ are the orthonormal bases (obtained by SVD or QR factorization) of $E_X$ and $E_Y$, respectively.
    \item[2.] \textbf{Update core matrices}. To align dimensions, we first pad $P$, $Q_X$, and $Q_Y$ with zeros on the dimensions where $U_X$ and $U_Y$ expanded. Then update core matrices by Eq.~\eqref{eq:update P}--\eqref{eq:update Q_Y}.
    \item[3.] \textbf{Compression}. The first two steps incorporate all new information at the cost of expanding dimensions. Next, we compress the model based on Theorem~\ref{p2}. Denote $V_X$ and $V_Y$ to be matrices composed by the leading $r_X$ and $r_Y$ eigenvectors of $Q_X$ and $Q_Y$, respectively. We can compress projection matrices by $U_X\leftarrow U_X V_X$, $U_Y\leftarrow U_Y V_Y$ to keep the leading singular vectors of $\bar{X}_{t+d}^w$ and $\bar{Y}_{t+d}^w$. The core matrices can be compressed accordingly by $Q_X\leftarrow V_X^{\top}Q_X V_X$, $Q_Y\leftarrow V_Y^{\top}Q_Y V_Y$, $P\leftarrow V_Y^{\top}P V_X$.
\end{itemize}

Besides the daily update, a more general setting can be updating the model for every $k$ intervals or only doing the compression step when $r_X$ or $r_Y$ exceeds a threshold. This paper adopts the daily update described above because metro systems often have a one-day periodicity. In terms of computational efficiency, the online update algorithm computes the SVD for $d$-column data matrices and eigenvalue decomposition of $Q_X$ and $Q_Y$. The computation has a constant cost every day and it is significantly faster than retraining using entire data. In terms of memory efficiency, historical data are not required when updating the model. All we need to store are three ``core'' matrices and two projection matrices. Regarding the error, the online algorithm does not take into account the previously truncated part. This impact is negligible because the truncated part contains mostly noise, and past data are forgotten exponentially. Our experiments in section~\ref{sec:effect_online} show the online algorithm performs pretty close to or even slightly better than retraining.

\subsection{Connections with other DMD models}
The proposed HW-DMD is closely related to Hankel-DMD \citep{brunton2017chaos, arbabi2017ergodic, avila2020data} and DMD with control \citep[DMDc,][]{proctor2016dynamic}. Hankel-DMD uses Hankel data matrices as input and output to model a non-linear dynamical system by a linear model; its DMD modes approximate to the Koopman modes. There is another model also named Higher Order DMD \citep[HODMD,][]{le2017higher}, which requires Hankelizing data in its estimation and is essentially similar to Hankel-DMD. Instead, the proposed HW-DMD uses raw snapshots as the output (the left side of Eq.~\eqref{eq:var matrix}) without using the Hankel structure. This formula is equivalent to a high-order vector autoregression model, which is neater and more suitable in the context of forecasting. Moreover, our model can use non-continuous orders and external variables (e.g., the boarding flow). Essentially, the external variables of our model can be regarded as the control term of a DMDc model.

The three-step online update algorithm for HW-DMD in this paper inherits from the work of \citet{hemati2014dynamic}. The original algorithm was developed for the exact DMD introduced in Section~\ref{sec:DMD}. Besides, the online DMD proposed by \citep{zhang2019online} considers the decaying weight of data, but the constant projection matrix in their assumption restricts the update effect. \citep{alfatlawi2020incremental} proposed an online algorithm for weighted DMD using incremental SVD, which is a different technique from our method. Our contribution is extending the algorithm proposed by \citet{hemati2014dynamic} to a high-order weighted version with the consideration of external regression covariates.

\section{Experiments}\label{sec:experiments}
In this section, we compare the proposed HW-DMD with other forecasting models using real-world data. We begin with an introduction to data and experimental settings. Next, we compare model performances by forecasting the OD matrices and the boarding flow derived from the OD matrices. Finally, we examine the long-term effect of the online HW-DMD update algorithm. The code for experiments is available from \url{https://github.com/mcgill-smart-transport/high-order-weighted-DMD}.

\subsection{Data and experimental settings}
We examine HW-DMD using the metro smart card data from two cities, Guangzhou and Hangzhou. Both data sets record the origin, destination, and entry and exit time of each metro trip. We focus on the forecast of workdays and connect each Friday to the next Monday. Details of the two data sets are as follows:
\begin{itemize}
    \item Guangzhou metro data: This data set covers around 301 million trips among 159 metro stations in Guangzhou from July 1st to Sept 30th, 2017. Guangzhou metro operates from 6:00 to 24:00. We use the first twenty weekdays (July 3rd to July 28th) as the training set, the following ten weekdays (July 31st to Aug 11th) as the validation set, and the following ten weekdays (Aug 14th to Aug 25th) as the test set. There are additional one-month data after the test set; we use these data to study the long-term effect of the online HW-DMD update algorithm.
    \item Hangzhou metro data\footnote{\url{https://doi.org/10.5281/zenodo.3145404}}: This is an open data set that covers 80 effective stations of Hangzhou metro from Jan 1st to Jan 25th, 2019. The operation hours are from 5:30 to 23:30. We use the first ten weekdays (Jan 1 to Jan 14) for training, the following four weekdays (Jan 15 to Jan 18) for validation, and the rest five weekdays (Jan 21 to Jan 25) for testing.
\end{itemize}

We aggregate OD snapshots by a 30-minute time interval, which means 36 snapshots per day for both cities. Note that a small interval may result in sparse OD matrices; we choose the 30-minute interval to balance the practical requirements. Figure~\ref{fig:histogram} shows the distribution of $o_{i,j}$ from an OD snapshot of a typical morning peak in Guangzhou. The distribution roughly follows a power law, with most OD pairs having small volumes while a few of them are significantly larger. The highly skewed distribution is very difficult to be properly handled by conventional forecasting models.

\begin{figure}
\begin{center}
\includegraphics[]{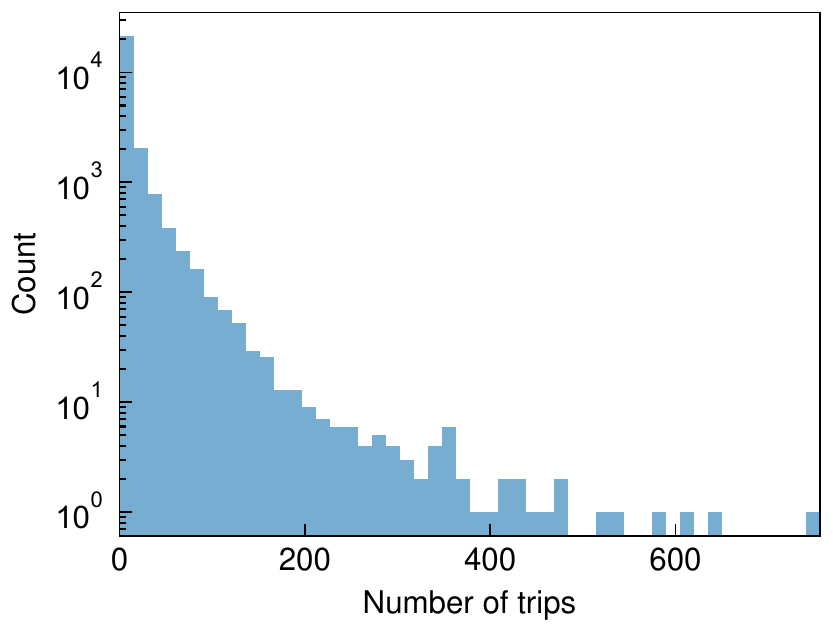}
\caption{The histogram of $o_{i,j}$ in an OD snapshot of a morning peak in Guangzhou.}
\label{fig:histogram}
\end{center}
\end{figure}

The performance of a model is quantified using the root-mean-square error (RMSE), the weighted mean absolute percentage error (WMAPE), and the coefficient of determination (denoted as $R^2$):
\begin{equation*}\label{eq:rmse}
    \text{RMSE}(\alpha, \hat{\alpha}) = \sqrt{\frac{1}{N} \sum_{i=1}^{N}{(\alpha_{i} - \hat{\alpha}_{i})^2}},
\end{equation*}
\begin{equation*}\label{eq:wmape}
    \text{WMAPE}(\alpha, \hat{\alpha}) =
    \frac{\sum_{i}^{N}\left| \alpha_{i} - \hat{\alpha}_{i}\right|}
    {\sum_{i}^{N} \left| \alpha_{i}\right|} \times 100\%,
\end{equation*}
\begin{equation*}\label{eq:r2}
    R^{2}(\alpha, \hat{\alpha})=1-\frac{\sum_{i=1}^{N}\left(\alpha_{i}-\hat{\alpha}_{i}\right)^{2}}{\sum_{i=1}^{N}\left(\alpha_{i}-\bar{\alpha}\right)^{2}},
\end{equation*}
where $\alpha$ and $\hat{\alpha}$ are, respectively, the real and predicted values; $\bar{\alpha}$ is the average value of $\alpha$; $N$ is the total number of elements under different time intervals and locations. The three performance metrics are computed for both OD flow $o$ and boarding flow $b$ (forecasted by $\hat{b}_{t,i}=\sum_{j}{\hat{o}_{t,i,j}}$).

\subsection{Hyperparameters}\label{sec:hyperparameters}
We use the online update algorithm for HW-DMD if not otherwise specified. Hyperparameters for HW-DMD include time lags $q_1, \cdots, q_h$, the SVD truncation rank $r_X, r_Y$, and the forgetting ratio $\rho$. These parameters are determined in a sequential order.

We use the Guangzhou data set as an example to elaborate the hyperparameter tuning procedure. We first set $r_X=r_Y=100$ and $\rho=1$ and select time lags in a greedy manner. For time lags within one day ($3\leq q_i \leq36$), we repeatedly add a ``currently best'' time lag based on the RMSE of the validation set until a new lag brings no improvement or the number of lags reaches ten. This procedure selects \{3, 4, 8, 14, 19, 28, 30, 33, 35, 36\} as time lags. The considerable high-order time lags in the result indicate long-term auto-correlations of OD time series. For example, the lag 19 roughly equals a typical work duration (9.5 hours), which can be explained as a strong correlation between the departure trips for commuters in the morning and the returning trips in the afternoon \citep{cheng2021incorporating}. The metro OD flow is also highly regular; the largest several lags (e.g., 33, 35, and 36) capture the one-day periodicity. Next, we determine $r_X$ and $r_Y$ by a grid search from 20 to 100 at an interval of 10. The best $r_Y$ is 50. A larger $r_X$ than 100 still brings a marginal improvement, but we truncate $r_X$ at 100 to restrict the model size ($r_X$ affects the size of $U_X$ in the online update). Lastly, we set $\rho$ to be $0.92$ based on a line search from $0.8$ to $1$. As shown in Figure~\ref{fig:rho}, we can see assigning smaller weights for old data indeed improves the forecast. Because $0.92^{8}\approx 0.51$, using $\rho=0.92$ is roughly equivalent to halving the weight every eight days.

\begin{figure}
\begin{center}
\includegraphics[]{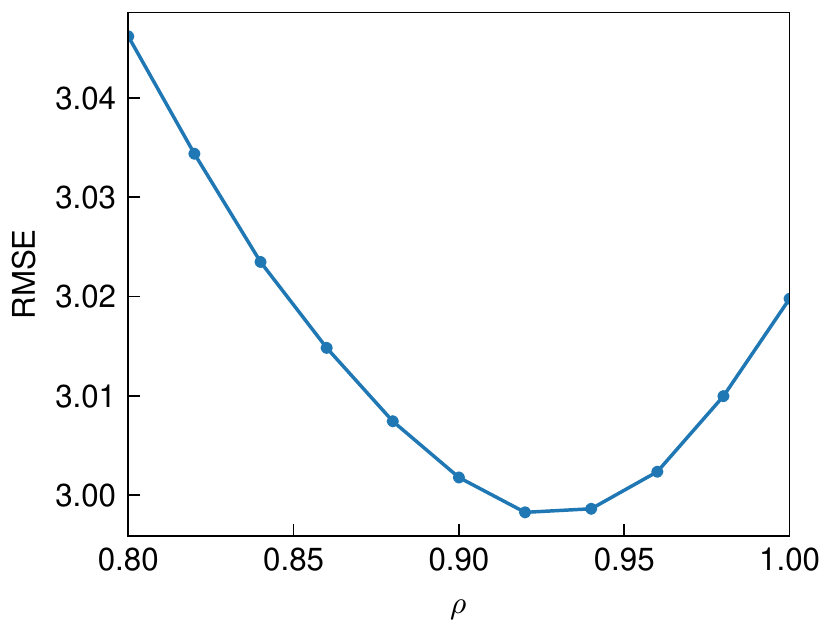}
\caption{The effect of $\rho$ to the forecast OD RMSE in the validation set of the Guangzhou data.}\label{fig:rho}
\end{center}
\end{figure}

The hyperparameter tuning for the Hangzhou data set follows the same procedure. The selected hyperparameters for the Hangzhou data set are time lags=\{3, 4, 6, 14, 18, 19, 28, 32, 35, 36\}, $r_Y=40$, $r_X=100$, and $\rho=0.92$.

\subsection{Benchmark models}
We compare HW-DMD with the following benchmark models:
\begin{itemize}
    \item HA: Historical Average. For the OD flow at a certain period (e.g., 7:00--7:30) of the day, HA uses the average OD flow at that period in the training set as the forecast value.

    \item TRMF: Temporal Regularized Matrix Factorization \citep{yu2016temporal}. TRMF is a matrix factorization model that imposes autoregression (AR) processes on each temporal factor. We use time lags $[1, \cdots ,10]$ for the AR processes. We search over \{100, 300, 500, 1000, 1500, 2000, 2500, 3000\} for the best regularization parameter and search from 30 to 150 with an interval of 10 for the best number of factors.

    \item ConvLSTM: Convolutional LSTM \citep{shi2015convolutional}. It is a deep recurrent neural network model that forecasts future frames of matrix time series (e.g., videos). Here we use it to forecast future OD matrices by the most recent ten OD matrices. Following the work by \cite{zhang2021short}, we apply a three-layer LSTM structure with eight, eight, and one filter, respectively, and set the kernel size to be $3\times3$ for all convolutional layers in the model.

    \item FNN: A two-layer Feedforward Neural Network. We use the OD snapshots of 3-10 lags ago and the boarding flow snapshot of 1-2 lags ago as the input features. We perform a grid search over the type of activation functions (linear, sigmoid, and relu) and the number of hidden layers (from 10 to 100 at an interval of 10) for the best model setting.

    \item SARIMA: Seasonal AutoRegressive Integrated Moving Average. We only use SARI-MA to forecast the boarding flow since SARIMA only handles one-dimensional time series. We use the order $\text{ARIMA}(2, 0, 1)(1, 1, 0)[36]$ for all the stations and fit 159 separate models. This model configuration is the same as \cite{cheng2021incorporating} and was tested to be suitable for most metro stations.
\end{itemize}

Applying TRMF, ConvLSTM, and FNN to the original data (or after a normalization) can hardly obtain a forecast better than HA. This phenomenon was also found by \cite{gong2018network, gong2020online}. This is because the OD data are high-dimensional, sparse, noisy, and highly skewed. To improve the forecast of these models, we apply TRMF, ConvLSTM, and FNN to the residuals after subtracting the HA from the original data.  This ``mean-removal'' processing also weakens the data's periodicity; therefore, we do not use seasonal lags in these models. Besides, because the standard TRMF and ConvLSTM cannot use the boarding flow as extra inputs, we ignore the delayed data availability problem for these models and assume all historical OD snapshots are available.

\subsection{Forecast result}
We apply trained models to the test set and forecast OD matrices of the next three steps at each time interval. Note OD snapshots of 1-2 lags ago are unknown; they are replaced by previously forecasted OD snapshots when doing multi-step rolling forecasting by HW-DMD/FNN. Next, the boarding flow can be calculated from OD matrices. We compare the forecast accuracy of models in terms of OD flow and boarding flow.

Table~\ref{tab:result_od} shows the results of OD flow forecast. We can see HW-DMD with a forgetting ratio $\rho=0.92$ outperforms other models in all evaluation metrics. Even the three-step forecast of HW-DMD is better than the one-step forecast of other models. The advantage of HW-DMD over other models is more significant in the Hangzhou data set. Although TRMF, FNN, and ConvLSTM are trained on the residuals after subtracting the HA from the original data, the improvement of these models compared with HA is limited. In contrast, HW-DMD is directly applied to the original data but provides a significantly better forecast, demonstrating its strong prediction power in handling the sparse, noisy, and high-dimensional OD data. Besides, the performance of the ``unweighted'' HW-DMD ($\rho=1$) is slightly behind the weighted version, but still better than other models.

\begin{table}[!htbp]
    \centering\small
    \caption{Models' performance for OD flow forecasting.}
    \label{tab:result_od}
    \begin{tabular}{ccccc|ccc}
    \toprule
    \multirowcell{2}{Method} &  \multirowcell{2}{Criterion}  & \multicolumn{3}{c|}{Guangzhou} & \multicolumn{3}{c}{Hangzhou}\\
        &       & One-step   & Two-step   & Three-step & One-step   & Two-step   & Three-step \\
    \hline
    \multirowcell{3}{HW-DMD \\ $\rho=0.92$} & RMSE  & \textbf{3.05}  & \textbf{3.09}  & \textbf{3.11}  & \textbf{3.36}  & \textbf{3.41}  & \textbf{3.44} \\
        & WMAPE & \textbf{29.65\%} & \textbf{29.77\%} & \textbf{29.79\%} & \textbf{31.76\%} & \textbf{31.96\%} & \textbf{31.84\%} \\
        & $R^2$   & \textbf{0.957} & \textbf{0.956} & \textbf{0.955} & \textbf{0.934} & \textbf{0.932} & \textbf{0.931} \\
    \hline
    \multirowcell{3}{HW-DMD \\ $\rho=1$} & RMSE  & 3.08  & 3.12  & 3.14  & 3.40  & 3.45  & 3.48 \\
        & WMAPE & 29.71\% & 29.87\% & 29.91\% & 31.94\% & 32.22\% & 32.13\% \\
        & $R^2$   & 0.956 & 0.955 & 0.954 & 0.933 & 0.930 & 0.929 \\
    \hline
    \multirow{3}[2]{*}{TRMF} & RMSE  & 3.22  & 3.24  & 3.26  & 3.80  & 3.89  & 3.96\\
        & WMAPE & 30.61\% & 30.72\% & 30.79\% & 34.02\% & 34.48\% & 34.82\% \\
        & $R^2$   & 0.952 & 0.951 & 0.951 & 0.916 & 0.912 & 0.908 \\
    \hline
    \multirow{3}[2]{*}{FNN} & RMSE  & 3.15  & 3.16  & 3.18  & 3.97  & 4.01  & 4.05\\
        & WMAPE & 30.23\% & 30.28\% & 30.32\% & 33.58\% & 33.63\% & 33.65\% \\
        & $R^2$   & 0.954 & 0.953 & 0.953 & 0.908 & 0.906 & 0.904 \\
    \hline
    \multirow{3}[2]{*}{Conv-LSTM} & RMSE  & 3.25  & 3.26  & 3.27  & 4.04  & 4.06  & 4.08\\
        & WMAPE & 30.11\% & 30.18\% & 30.23\% & 32.96\% & 32.92\% & 33.04\% \\
        & $R^2$   & 0.951 & 0.950 & 0.950 & 0.905 & 0.904 & 0.903 \\
    \hline
    \multirow{3}[2]{*}{HA} & RMSE  & 3.43  & 3.43  & 3.43  & 4.34  & 4.34  & 4.34\\
        & WMAPE & 31.21\% & 31.21\% & 31.21\% & 34.28\% & 34.28\% & 34.28\% \\
        & $R^2$   & 0.945 & 0.945 & 0.945 & 0.890 & 0.890 & 0.890  \\
    \bottomrule
    \end{tabular}
\end{table}%

\begin{table}[!htbp]
    \centering\small
    \caption{Models' performance for boarding flow forecasting.}\label{tab:result_boarding}
    \begin{tabular}{ccccc|ccc}
    \toprule
    \multirowcell{2}{Method} &  \multirowcell{2}{Criterion}  & \multicolumn{3}{c|}{Guangzhou} & \multicolumn{3}{c}{Hangzhou}\\
            &       & One-step   & Two-step   & Three-step & One-step   & Two-step   & Three-step \\
    \hline
    \multirowcell{3}{HW-DMD \\ $\rho=0.92$} & RMSE  & \textbf{93.99} & 102.61 & 107.58 & \textbf{50.08} & \textbf{54.14} & \textbf{56.32}\\
            & WMAPE & \textbf{6.09\%} & \textbf{6.68\%} & 6.98\% & \textbf{7.38\%} & \textbf{8.05\%} & \textbf{8.12\%} \\
            & $R^2$   & \textbf{0.991} & \textbf{0.989} & 0.988 & \textbf{0.989} & \textbf{0.988} & \textbf{0.987} \\
    \hline
    \multirowcell{3}{HW-DMD \\ $\rho=1$} & RMSE  & 94.51 & \textbf{102.46} & 106.55 & 51.28 & 55.66 & 58.45\\
            & WMAPE & 6.18\% & 6.74\% & 6.98\% & 7.54\% & 8.29\% & 8.43\% \\
            & $R^2$   & 0.991 & 0.989 & 0.988 & 0.989 & 0.987 & 0.986 \\
    \hline
    \multirow{3}[2]{*}{TRMF} & RMSE  & 126.03 & 127.87 & 128.65 & 77.70 & 81.19 & 83.12 \\
            & WMAPE & 7.92\% & 8.07\% & 8.13\% & 10.00\% & 10.55\% & 10.81\% \\
            & $R^2$   & 0.983 & 0.983 & 0.983 & 0.975 & 0.972 & 0.971 \\
    \hline
    \multirow{3}[2]{*}{FNN} & RMSE  & 101.93 & 104.00 & \textbf{106.06} & 67.16 & 68.83 & 70.77 \\
            & WMAPE & 6.44\% & 6.58\% & \textbf{6.69}\% & 9.00\% & 9.22\% & 9.50\% \\
            & $R^2$   & 0.989 & 0.989 & \textbf{0.988} & 0.981 & 0.980 & 0.979 \\
    \hline
    \multirow{3}[2]{*}{Conv-LSTM} & RMSE  & 117.16 & 121.22 & 123.40 & 71.46 & 75.75 & 78.07 \\
            & WMAPE & 6.87\% & 7.19\% & 7.35\% & 8.83\% & 9.63\% & 9.98\% \\
            & $R^2$   & 0.985 & 0.984 & 0.984 & 0.978 & 0.976 & 0.974 \\
    \hline
    \multirow{3}[2]{*}{HA} & RMSE  & 136.56 & 136.56 & 136.56 & 88.25 & 88.25 & 88.25 \\
            & WMAPE & 8.38\% & 8.38\% & 8.38\% & 11.09\% & 11.09\% & 11.09\% \\
            & $R^2$   & 0.980 & 0.980 & 0.980 & 0.967 & 0.967 & 0.967 \\
    \hline
    \multirow{3}[2]{*}{SARIMA} & RMSE  & 110.23 & 120.60 & 126.52 & 55.59 & 60.97 & 64.66 \\
            & WMAPE & 7.15\% & 7.65\% & 7.93\% & 7.86\% & 8.28\% & 8.50\% \\
            & $R^2$   & 0.987 & 0.985 & 0.983 & 0.987 & 0.984 & 0.982 \\
    \bottomrule
    \end{tabular}%
\end{table}%

Examining the aggregated boarding flow is important because it reflects if the forecast errors in OD matrices' are properly distributed, which is crucial when using OD matrices in traffic assignments. Moreover, the boarding flow itself is of interest to many applications. Table~\ref{tab:result_boarding} shows the boarding flow forecasting; all models except SARIMA calculate boarding flow by OD matrices. The two HW-DMD models are the best models in most cases. The only exception is that FNN slightly outperforms HW-DMD for the three-step forecast of the Guangzhou data set. Importantly, HW-DMD is the only model that outperforms SARIMA, a well-established boarding flow forecasting model, in both data sets, showing that the forecast of HW-DMD accurately reflects the marginal distribution of OD matrices.

The magnitude of OD flow in a metro system varies significantly in time and space dimensions. Therefore, we further compare HW-DMD with other models under different scenarios. Figure~\ref{fig:error_distribution} (a) and (c) show the forecast RMSE at different times of a day. We can see the RMSE of HW-DMD is the smallest in most time slots, particularly for the Hangzhou data set. Other models, such as Conv-LSTM, perform slightly better in the early morning and late night, but the difference is close, and the total network OD flow of these periods is pretty small. Figure~\ref{fig:error_distribution} (b) and (d) show the forecast RMSE for OD pairs with different flow magnitudes. The forecast RMSE of HW-DMD is considerably lower than other models for high-flow OD pairs (average half-hour OD flow larger than $2^4$). Note the number of OD pairs drops exponentially with the increase of OD flow, showing the superior forecast capability of HW-DMD for highly skewed data.

\begin{figure}[ht]
\begin{center}
\includegraphics[]{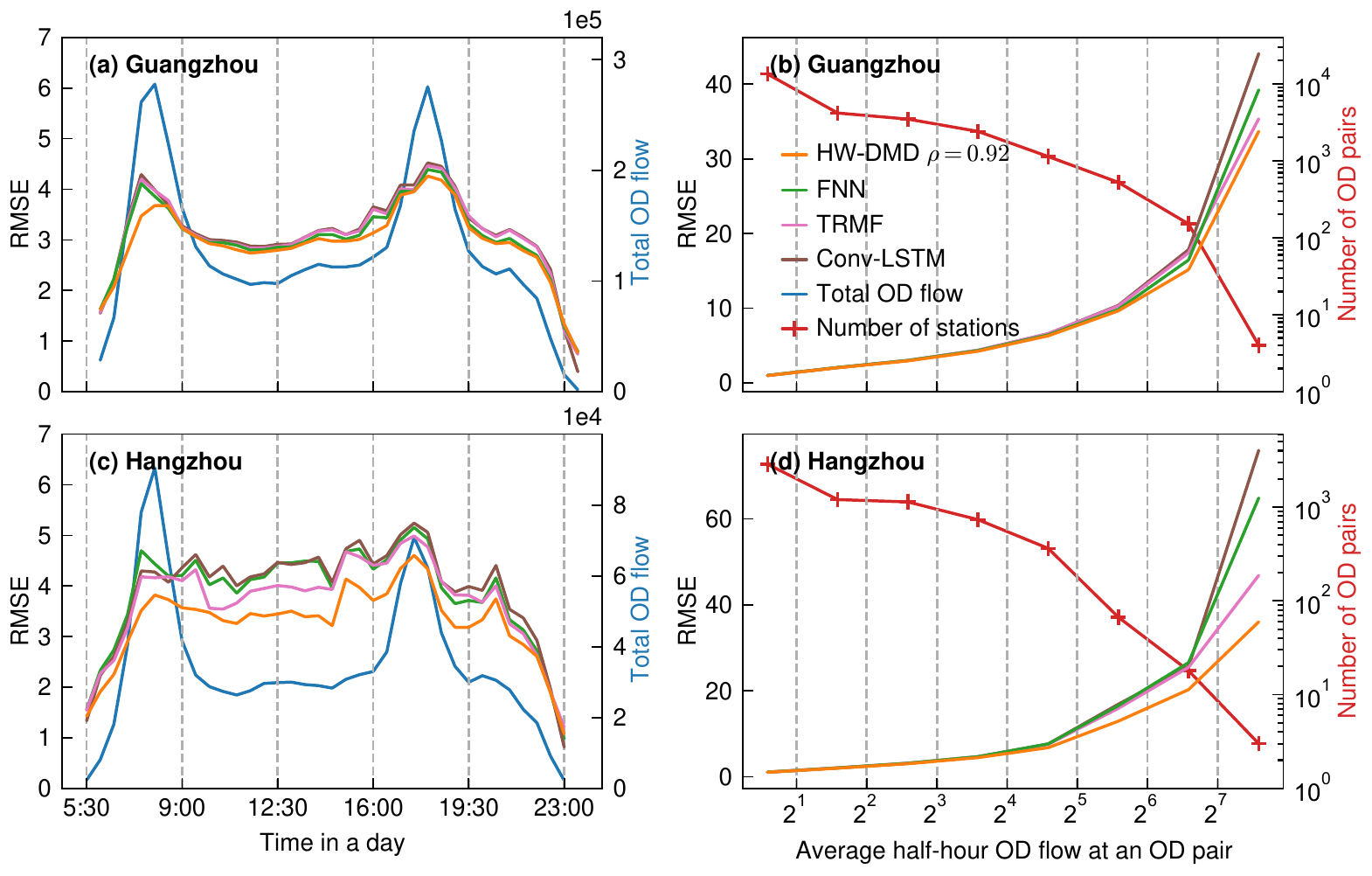}
\caption{The RMSE of OD flow forecasting at different times and different OD pairs. (a) and (c) shows the RMSE of OD matrix forecasting at every 30-minute interval, along with the total OD flow in the network. Using $2^i$ as boundaries, we divide OD pairs into groups according to their average half-hour OD flow; the forecast RMSE at each group and the number of OD pairs of each group are shown in (b) and (d).}\label{fig:error_distribution}
\end{center}
\end{figure}


Finally, we show the real and one-step forecast of OD flow at four representative OD pairs of Guangzhou metro in  Figure~\ref{fig:forecast}. The OD flow exhibits a clear daily periodicity, explaining why HA already works reasonably well. Compared with FNN, HW-DMD is better at forecasting the fluctuation of high-flow OD pairs, as shown in Figure~\ref{fig:forecast} (a) and (b). In Figure~\ref{fig:forecast} (a), the forecast of HW-DMD is often lower than the real value; this is hard to avoid since there is a two-lag delay when collecting the real OD flow. More OD pairs in the system are like Figure~\ref{fig:forecast} (c) and (d) with a low flow but high noise. Under such high volatility, the forecast by HW-DMD reflects a smooth average value. In fact, the performances of other models are often undermined by noise. The SVD truncation to the data greatly enhances HW-DMD's ability in handling the noise data (Figure~\ref{fig:singular value}). Overall, HW-DMD achieves a great balance between forecasting and noise reduction, which is particularly hard for such a high-dimensional system with diverse flow magnitudes.

\begin{figure}[ht]
\begin{center}
\includegraphics[]{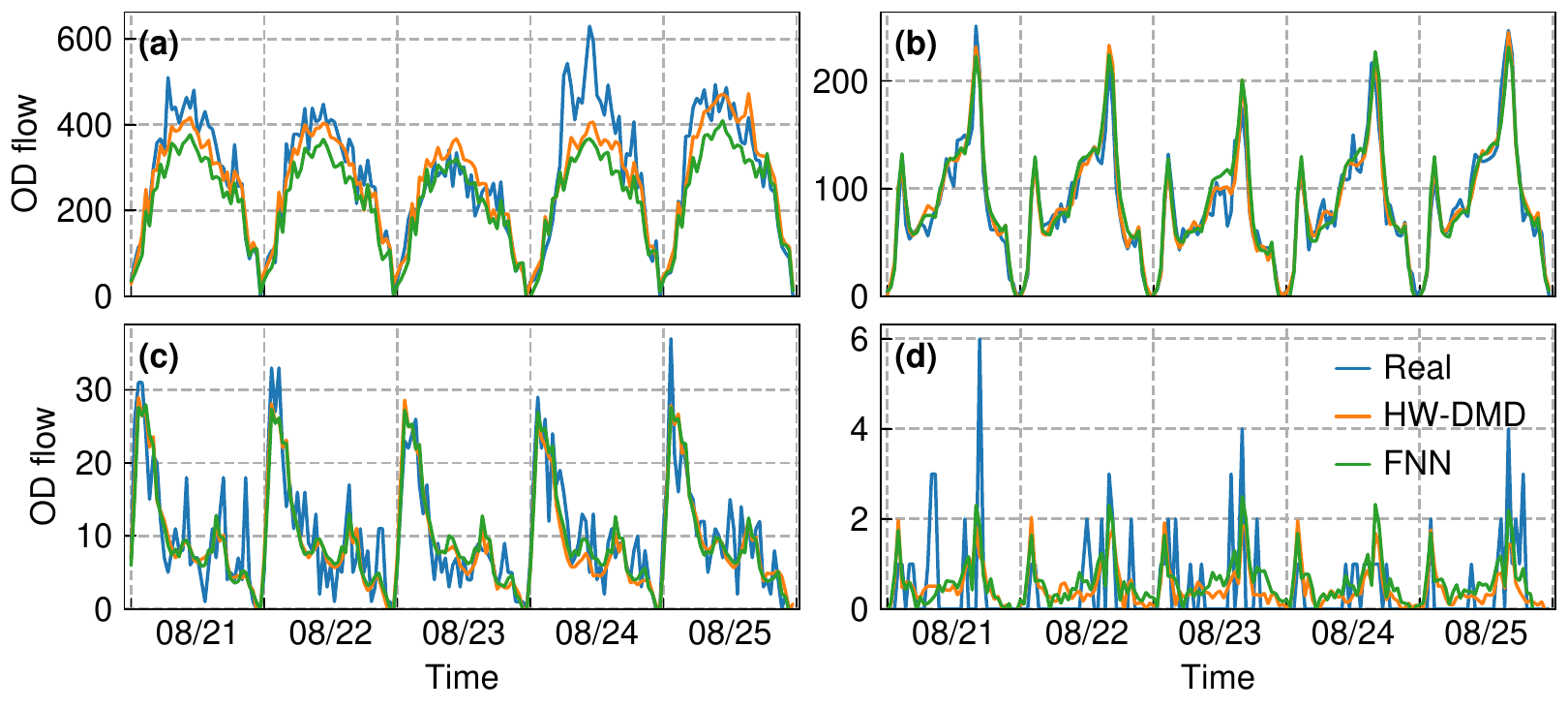}
\caption{The real and forecasted time series of four selected OD pairs of Guangzhou metro. (a) is the busiest OD pair in the Guangzhou metro data set. (a) to (d) are in a flow decreasing order.}\label{fig:forecast}
\end{center}
\end{figure}

\subsection{Effect of the low-rank assumption}\label{sec:impact_rank}
The demands of majority OD pairs are small and sparse by nature, making it difficult for a forecasting model to distinguish random fluctuation (noise) and intrinsic dynamic patterns. Taking the OD pair shown in Figure~\ref{fig:forecast} (d) as an example, the randomness in this OD pair is quite large compared with its average flow (low signal-to-noise ratio). A good forecasting should be robust to the noise while maintaining accurate cumulative effects of OD pairs in total (e.g., the boarding flow). This section evaluates the impact of using the low-rank assumption on forecasting and noise filtering.

According to Section~\ref{sec:estimation}, the forecast of HW-DMD is always on the column space of $U_Y$. Therefore, the best possible value of an OD snapshot $\hat{\mathbf{y}}_i$ calculated by HW-DMD is the rank-reduced full-size data, i.e., $U_Y U_Y^{\top}\mathbf{y}_i$, which is the upper bound of an HW-DMD's forecast ability. Figure~\ref{fig:low_rank_approx} shows how well this low-rank approximation fits the original data. We can see the low-rank approximation keeps most information for the high-demand OD pair of Figure~\ref{fig:low_rank_approx} (a). In contrast, most fluctuations in the sparse-demand OD pair of Figure~\ref{fig:low_rank_approx} (b) are truncated. By comparing with HA, we can see the low-rank approximation reflects the average daily pattern of the sparse-demand OD pair, which is a reasonable approximation when considering the cumulative effects of OD pairs. Therefore, the rank truncation is crucial for filtering the noise in a large number of sparse-demand OD pairs.

\begin{figure}[ht]
\begin{center}
\includegraphics[]{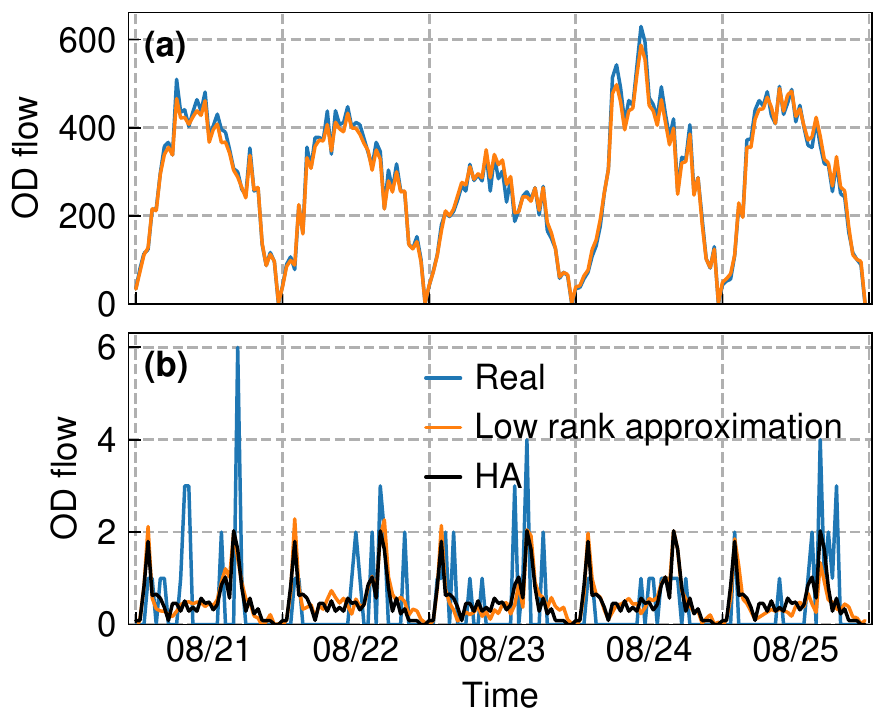}
\caption{Comparing OD flow with its low-rank approximation in two Guangzhou metro OD pairs. (a) and (b) corresponds to the (a) and (d) in Figure~\ref{fig:forecast}, respectively.}\label{fig:low_rank_approx}
\end{center}
\end{figure}

\begin{table}[htbp]
    \centering\small
    \caption{The difference between original data and the low-rank approximation.}
    \label{tab:result_lowrank}
    \begin{tabular}{ccc|c}
      \toprule
        Variable & Criterion & Guangzhou & Hangzhou \\
      \hline
      \multirow{3}[2]{*}{OD flow} & RMSE  & 2.82  & 3.00 \\
            & WMAPE & 28.80\% & 30.65\% \\
            & $R^2$   & 0.963 & 0.947 \\
      \hline
      \multirow{3}[2]{*}{Boarding flow} & RMSE  & 64.15 & 36.89 \\
            & WMAPE & 4.69\% & 5.95\% \\
            & $R^2$   & 0.996 & 0.994 \\
    \bottomrule
\end{tabular}%
\end{table}%

Table~\ref{tab:result_lowrank} further quantitatively evaluates the differences between the original OD data and its low-rank approximation. The results in Table~\ref{tab:result_lowrank} are the forecast upper bound of HW-DMD under the current rank-reduced space. By comparing Table~\ref{tab:result_lowrank} with the forecast of HW-DMD in Table~\ref{tab:result_od} and Table~\ref{tab:result_boarding}, we can see that a significant portion of the forecast error of HW-DMD essentially attributes to the rank truncation, but there is still space to improve the current HW-DMD model (e.g., by higher order, larger $r_X$, more regression covariates).

In choosing the rank-reduced space, the two rank parameters in HW-DMD balance the trade-off between forecast accuracy and model complexity. Based on the results of hyperparameter tuning, a further increase in rank $r_Y$ may result in overfitting (bringing the noise into the rank-reduced target data). We can further slightly improve the forecasting accuracy of HW-DMD by increasing the rank $r_X$ (related to the rank-reduced input data), but we here prefer a compact model with a smaller $r_X$ at the cost of slight accuracy loss. Lastly, the current HW-DMD chooses the rank-reduced space purely based on the leading singular values, which may be sensitive and not optimal when encountering significant data anomalies and failures \citep{duke2012error}. Using optimized DMD \citep{chen2012variants} or combined with an anomaly detection algorithm \citep{scherl2020robust} could further improve the current HW-DMD.

\subsection{Effect of the online update}\label{sec:effect_online}
The online update algorithm proposed in Section~\ref{sec:online} can update HW-DMD's parameters daily without storing historical data, which is computationally more efficient. On the Guangzhou metro training set, it takes $18.7\pm0.43$ seconds to train an HW-DMD model, while the online update only takes $1.0\pm0.03$ seconds for each day\footnote{We report the mean $\pm$ standard deviation of 20 runs. Tests were run on a computer with Intel$^\text{\tiny{\textregistered}}$ Core$^\text{TM}$ i7-8700 Processor and 24 gigabytes of RAM. Other benchmark models have much longer training time than HW-DMD (more than 1 minute for FNN and more than 20 minutes for TRMF and Conv-LSTM).} Besides the training time, we particularly care about if errors will accumulate if we keep using the online update algorithm for a long time. Therefore, we apply the online update algorithm to all the two-month data after the training set of the Guangzhou data set to evaluate its long-term effect. In comparison, we retrain two HW-DMD models (with $\rho$=0.92 and 1, respectively) every day using all historical data up until the latest. The results are shown in Figure~\ref{fig:longterm_online}. We summarize the key findings for Figure~\ref{fig:longterm_online} as follows:
\begin{itemize}
    \item The RMSE of a constant model gradually increases over time. This indicates the metro system's dynamics are time-evolving; thus, forecasting models should be updated/retrained regularly for better performance.

    \item The RMSE curve of the online update algorithm clings to the model ($\rho=0.92$) retrained every day by entire historical data, showing the online HW-DMD update algorithm works consistently well in long-term applications. For a large training set (e.g., after September in Figure~\ref{fig:longterm_online}), the online update approach even performs slightly better than retraining.

    \item Properly reducing the weight for old data improves the forecast. Compare $\rho=0.92$ with $\rho=1$ for the two retrained models; the benefits of forgetting the old data become more significant as the training data increases.

    \item The OD flow of certain weekdays can be harder to forecast. Especially for the forecast of September. The RMSEs on Fridays are significantly higher on than other weekdays.
\end{itemize}

\begin{figure}
\begin{center}
\includegraphics[]{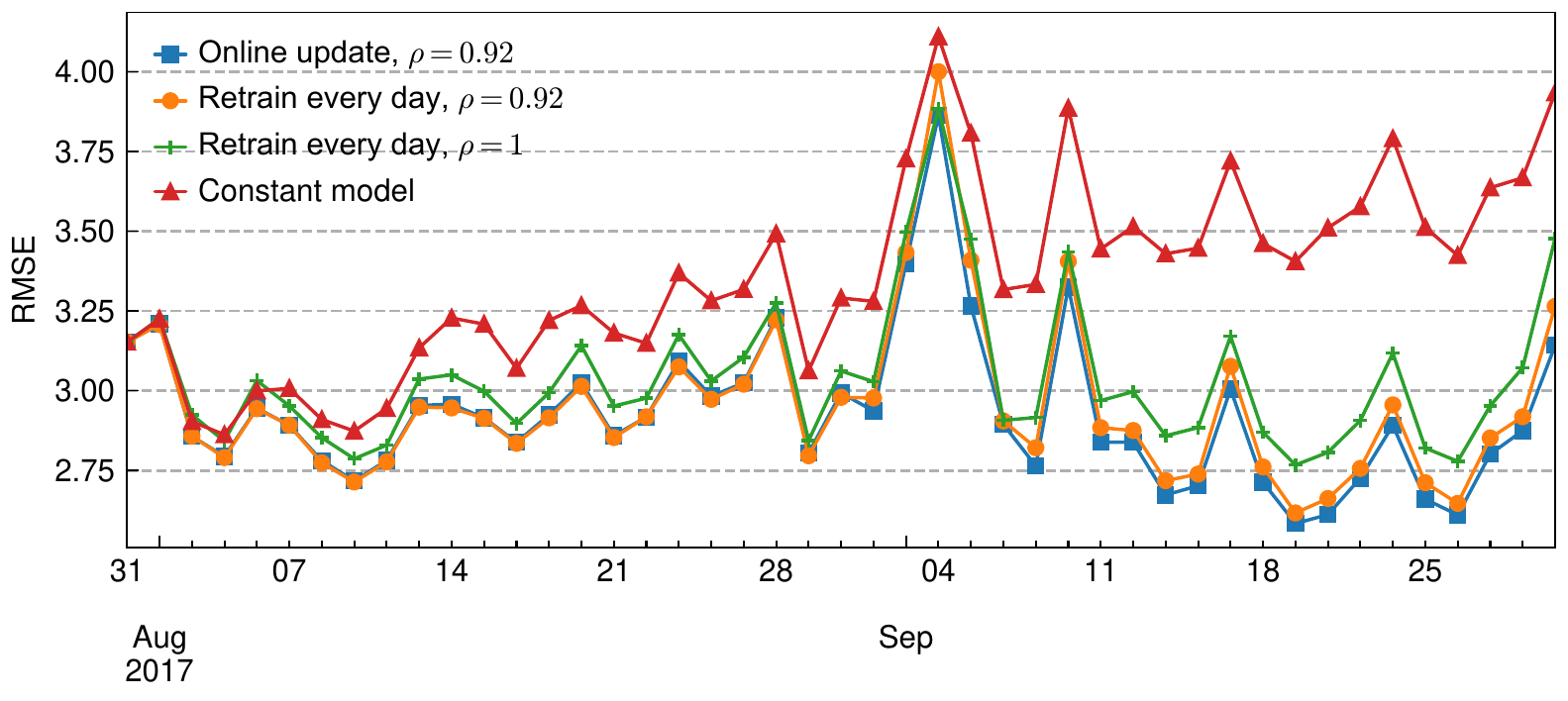}
\caption{The evolvement of forecast RMSE under different HW-DMD update methods (Guangzhou data set). Each marker represents the RMSE of forecasted OD flow during a day. Numbered dates in the horizontal axis are Mondays; weekends are skipped.}
\label{fig:longterm_online}
\end{center}
\end{figure}

Many forecasting models do not consider the time-evolving dynamics of a metro system. Regular retraining can be prohibitive, especially for complicated models (e.g., deep learning models). This experiment shows the online update algorithm for HW-DMD is a memory-saving and accurate approach to keep an HW-DMD model up to date.

\section{Conclusions and Discussion}\label{sec:conclusions}
This paper proposes a high-order weighted dynamic mode decomposition (HW-DMD) model to solve the real-time short-term OD matrix forecasting problem in metro systems. Experiments show that HW-DMD significantly outperforms common forecasting models under the high-dimensional, sparse, noisy, and skewed OD data. Particularly, we address the delayed data availability problem and the time-evolving dynamics of metro systems, which are often ignored in the literature. The idea of the forgetting rate and online update in dealing with a time-evolving system is also beneficial for other forecasting models. Moreover, the implementation of HW-DMD is simple, and the computation is very efficient, providing a promising solution to general high-dimensional time series forecasting problems.

We discuss several future research directions. (1) Current HW-DMD reshapes OD matrices into vectors for dimensionality reduction. However, performing dimensionality reduction directly on OD matrices may better utilize the column/row-wise correlations and produce more concise models \citep{chen2020autoregressive, gong2020online}. A difficulty in this direction is that the low-rank feature in metro OD matrices is relatively weak because the diagonal elements of metro OD matrices are all zeros. (2) Another future direction is to use a non-linear model instead of the current linear model in the reduced space, such as the deep factor model \citep{wang2019deep}. But a limitation for a non-linear model is that an online update method may be difficult to derive or even  impossible. (3) Lastly, current HW-DMD uses external features, such as the boarding flow, simply as covariates. Incorporating more general features (e.g., weather, events) and network structure to improve the HW-DMD is worth investigating.

\section*{Acknowledgement}
This research is supported by the Natural Sciences and Engineering Research Council (NSERC) of Canada, Mitacs, exo.quebec (https://exo.quebec/en), and the Canada Foundation for Innovation (CFI). The code for HW-DMD is available at \url{https://github.com/mcgill-smart-transport/high-order-weighted-DMD}.

\bibliographystyle{elsarticle-harv}
\bibliography{ref}

\begin{thebibliography}{48}
\expandafter\ifx\csname natexlab\endcsname\relax\def\natexlab#1{#1}\fi
\providecommand{\url}[1]{\texttt{#1}}
\providecommand{\href}[2]{#2}
\providecommand{\path}[1]{#1}
\providecommand{\DOIprefix}{doi:}
\providecommand{\ArXivprefix}{arXiv:}
\providecommand{\URLprefix}{URL: }
\providecommand{\Pubmedprefix}{pmid:}
\providecommand{\doi}[1]{\href{http://dx.doi.org/#1}{\path{#1}}}
\providecommand{\Pubmed}[1]{\href{pmid:#1}{\path{#1}}}
\providecommand{\bibinfo}[2]{#2}
\ifx\xfnm\relax \def\xfnm[#1]{\unskip,\space#1}\fi
\bibitem[{Alfatlawi and Srivastava(2020)}]{alfatlawi2020incremental}
\bibinfo{author}{Alfatlawi, M.}, \bibinfo{author}{Srivastava, V.},
  \bibinfo{year}{2020}.
\newblock \bibinfo{title}{An incremental approach to online dynamic mode
  decomposition for time-varying systems with applications to eeg data
  modeling.}
\newblock \bibinfo{journal}{Journal of Computational Dynamics}
  \bibinfo{volume}{7}, \bibinfo{pages}{209--241}.
\bibitem[{Arbabi and Mezi{\'c}(2017)}]{arbabi2017ergodic}
\bibinfo{author}{Arbabi, H.}, \bibinfo{author}{Mezi{\'c}, I.},
  \bibinfo{year}{2017}.
\newblock \bibinfo{title}{Ergodic theory, dynamic mode decomposition, and
  computation of spectral properties of the koopman operator}.
\newblock \bibinfo{journal}{SIAM Journal on Applied Dynamical Systems}
  \bibinfo{volume}{16}, \bibinfo{pages}{2096--2126}.
\bibitem[{Avila and Mezi{\'c}(2020)}]{avila2020data}
\bibinfo{author}{Avila, A.}, \bibinfo{author}{Mezi{\'c}, I.},
  \bibinfo{year}{2020}.
\newblock \bibinfo{title}{Data-driven analysis and forecasting of highway
  traffic dynamics}.
\newblock \bibinfo{journal}{Nature communications} \bibinfo{volume}{11},
  \bibinfo{pages}{1--16}.
\bibitem[{Brunton et~al.(2017)Brunton, Brunton, Proctor, Kaiser and
  Kutz}]{brunton2017chaos}
\bibinfo{author}{Brunton, S.L.}, \bibinfo{author}{Brunton, B.W.},
  \bibinfo{author}{Proctor, J.L.}, \bibinfo{author}{Kaiser, E.},
  \bibinfo{author}{Kutz, J.N.}, \bibinfo{year}{2017}.
\newblock \bibinfo{title}{Chaos as an intermittently forced linear system}.
\newblock \bibinfo{journal}{Nature communications} \bibinfo{volume}{8},
  \bibinfo{pages}{1--9}.
\bibitem[{Chen et~al.(2019)Chen, Ye, Wang and Xu}]{chen2019subway}
\bibinfo{author}{Chen, E.}, \bibinfo{author}{Ye, Z.}, \bibinfo{author}{Wang,
  C.}, \bibinfo{author}{Xu, M.}, \bibinfo{year}{2019}.
\newblock \bibinfo{title}{Subway passenger flow prediction for special events
  using smart card data}.
\newblock \bibinfo{journal}{IEEE Transactions on Intelligent Transportation
  Systems} \bibinfo{volume}{21}, \bibinfo{pages}{1109--1120}.
\bibitem[{Chen et~al.(2012)Chen, Tu and Rowley}]{chen2012variants}
\bibinfo{author}{Chen, K.K.}, \bibinfo{author}{Tu, J.H.},
  \bibinfo{author}{Rowley, C.W.}, \bibinfo{year}{2012}.
\newblock \bibinfo{title}{Variants of dynamic mode decomposition: boundary
  condition, koopman, and fourier analyses}.
\newblock \bibinfo{journal}{Journal of nonlinear science} \bibinfo{volume}{22},
  \bibinfo{pages}{887--915}.
\bibitem[{Chen et~al.(2021)Chen, Xiao and Yang}]{chen2020autoregressive}
\bibinfo{author}{Chen, R.}, \bibinfo{author}{Xiao, H.}, \bibinfo{author}{Yang,
  D.}, \bibinfo{year}{2021}.
\newblock \bibinfo{title}{Autoregressive models for matrix-valued time series}.
\newblock \bibinfo{journal}{Journal of Econometrics} \bibinfo{volume}{222},
  \bibinfo{pages}{539--560}.
\bibitem[{Cheng et~al.(2021)Cheng, Tr{\'e}panier and
  Sun}]{cheng2021incorporating}
\bibinfo{author}{Cheng, Z.}, \bibinfo{author}{Tr{\'e}panier, M.},
  \bibinfo{author}{Sun, L.}, \bibinfo{year}{2021}.
\newblock \bibinfo{title}{Incorporating travel behavior regularity into
  passenger flow forecasting}.
\newblock \bibinfo{journal}{Transportation Research Part C: Emerging
  Technologies} \bibinfo{volume}{128}, \bibinfo{pages}{103200}.
\bibitem[{Chu et~al.(2019)Chu, Lam and Li}]{chu2019deep}
\bibinfo{author}{Chu, K.F.}, \bibinfo{author}{Lam, A.Y.}, \bibinfo{author}{Li,
  V.O.}, \bibinfo{year}{2019}.
\newblock \bibinfo{title}{Deep multi-scale convolutional lstm network for
  travel demand and origin-destination predictions}.
\newblock \bibinfo{journal}{IEEE Transactions on Intelligent Transportation
  Systems} .
\bibitem[{Dai et~al.(2018)Dai, Sun and Xu}]{dai2018short}
\bibinfo{author}{Dai, X.}, \bibinfo{author}{Sun, L.}, \bibinfo{author}{Xu, Y.},
  \bibinfo{year}{2018}.
\newblock \bibinfo{title}{Short-term origin-destination based metro flow
  prediction with probabilistic model selection approach}.
\newblock \bibinfo{journal}{Journal of Advanced Transportation}
  \bibinfo{volume}{2018}.
\bibitem[{Duke et~al.(2012)Duke, Soria and Honnery}]{duke2012error}
\bibinfo{author}{Duke, D.}, \bibinfo{author}{Soria, J.},
  \bibinfo{author}{Honnery, D.}, \bibinfo{year}{2012}.
\newblock \bibinfo{title}{An error analysis of the dynamic mode decomposition}.
\newblock \bibinfo{journal}{Experiments in fluids} \bibinfo{volume}{52},
  \bibinfo{pages}{529--542}.
\bibitem[{Eckart and Young(1936)}]{eckart1936approximation}
\bibinfo{author}{Eckart, C.}, \bibinfo{author}{Young, G.},
  \bibinfo{year}{1936}.
\newblock \bibinfo{title}{The approximation of one matrix by another of lower
  rank}.
\newblock \bibinfo{journal}{Psychometrika} \bibinfo{volume}{1},
  \bibinfo{pages}{211--218}.
\bibitem[{Forni et~al.(2000)Forni, Hallin, Lippi and
  Reichlin}]{forni2000generalized}
\bibinfo{author}{Forni, M.}, \bibinfo{author}{Hallin, M.},
  \bibinfo{author}{Lippi, M.}, \bibinfo{author}{Reichlin, L.},
  \bibinfo{year}{2000}.
\newblock \bibinfo{title}{The generalized dynamic-factor model: Identification
  and estimation}.
\newblock \bibinfo{journal}{Review of Economics and statistics}
  \bibinfo{volume}{82}, \bibinfo{pages}{540--554}.
\bibitem[{Gong et~al.(2020)Gong, Li, Zhang, Liu and Zheng}]{gong2020online}
\bibinfo{author}{Gong, Y.}, \bibinfo{author}{Li, Z.}, \bibinfo{author}{Zhang,
  J.}, \bibinfo{author}{Liu, W.}, \bibinfo{author}{Zheng, Y.},
  \bibinfo{year}{2020}.
\newblock \bibinfo{title}{Online spatio-temporal crowd flow distribution
  prediction for complex metro system}.
\newblock \bibinfo{journal}{IEEE Transactions on Knowledge and Data
  Engineering} .
\bibitem[{Gong et~al.(2018)Gong, Li, Zhang, Liu, Zheng and
  Kirsch}]{gong2018network}
\bibinfo{author}{Gong, Y.}, \bibinfo{author}{Li, Z.}, \bibinfo{author}{Zhang,
  J.}, \bibinfo{author}{Liu, W.}, \bibinfo{author}{Zheng, Y.},
  \bibinfo{author}{Kirsch, C.}, \bibinfo{year}{2018}.
\newblock \bibinfo{title}{Network-wide crowd flow prediction of sydney trains
  via customized online non-negative matrix factorization}, in:
  \bibinfo{booktitle}{Proceedings of the 27th ACM International Conference on
  Information and Knowledge Management}, pp. \bibinfo{pages}{1243--1252}.
\bibitem[{Hemati et~al.(2014)Hemati, Williams and Rowley}]{hemati2014dynamic}
\bibinfo{author}{Hemati, M.S.}, \bibinfo{author}{Williams, M.O.},
  \bibinfo{author}{Rowley, C.W.}, \bibinfo{year}{2014}.
\newblock \bibinfo{title}{Dynamic mode decomposition for large and streaming
  datasets}.
\newblock \bibinfo{journal}{Physics of Fluids} \bibinfo{volume}{26},
  \bibinfo{pages}{111701}.
\bibitem[{Hu et~al.(2020)Hu, Yang, Guo, Jensen and Xiong}]{hu2020stochastic}
\bibinfo{author}{Hu, J.}, \bibinfo{author}{Yang, B.}, \bibinfo{author}{Guo,
  C.}, \bibinfo{author}{Jensen, C.S.}, \bibinfo{author}{Xiong, H.},
  \bibinfo{year}{2020}.
\newblock \bibinfo{title}{Stochastic origin-destination matrix forecasting
  using dual-stage graph convolutional, recurrent neural networks}, in:
  \bibinfo{booktitle}{2020 IEEE 36th International Conference on Data
  Engineering (ICDE)}, \bibinfo{organization}{IEEE}. pp.
  \bibinfo{pages}{1417--1428}.
\bibitem[{Ke et~al.(2021)Ke, Qin, Yang, Zheng, Zhu and Ye}]{ke2021predicting}
\bibinfo{author}{Ke, J.}, \bibinfo{author}{Qin, X.}, \bibinfo{author}{Yang,
  H.}, \bibinfo{author}{Zheng, Z.}, \bibinfo{author}{Zhu, Z.},
  \bibinfo{author}{Ye, J.}, \bibinfo{year}{2021}.
\newblock \bibinfo{title}{Predicting origin-destination ride-sourcing demand
  with a spatio-temporal encoder-decoder residual multi-graph convolutional
  network}.
\newblock \bibinfo{journal}{Transportation Research Part C: Emerging
  Technologies} \bibinfo{volume}{122}, \bibinfo{pages}{102858}.
\bibitem[{Kwak and Geroliminis(2020)}]{kwak2020travel}
\bibinfo{author}{Kwak, S.}, \bibinfo{author}{Geroliminis, N.},
  \bibinfo{year}{2020}.
\newblock \bibinfo{title}{Travel time prediction for congested freeways with a
  dynamic linear model}.
\newblock \bibinfo{journal}{IEEE Transactions on Intelligent Transportation
  Systems} .
\bibitem[{Lam et~al.(2011)Lam, Yao and Bathia}]{lam2011estimation}
\bibinfo{author}{Lam, C.}, \bibinfo{author}{Yao, Q.}, \bibinfo{author}{Bathia,
  N.}, \bibinfo{year}{2011}.
\newblock \bibinfo{title}{Estimation of latent factors for high-dimensional
  time series}.
\newblock \bibinfo{journal}{Biometrika} \bibinfo{volume}{98},
  \bibinfo{pages}{901--918}.
\bibitem[{Le~Clainche and Vega(2017)}]{le2017higher}
\bibinfo{author}{Le~Clainche, S.}, \bibinfo{author}{Vega, J.M.},
  \bibinfo{year}{2017}.
\newblock \bibinfo{title}{Higher order dynamic mode decomposition}.
\newblock \bibinfo{journal}{SIAM Journal on Applied Dynamical Systems}
  \bibinfo{volume}{16}, \bibinfo{pages}{882--925}.
\bibitem[{Li et~al.(2017)Li, Wang, Sun, Ma and Lu}]{li2017forecasting}
\bibinfo{author}{Li, Y.}, \bibinfo{author}{Wang, X.}, \bibinfo{author}{Sun,
  S.}, \bibinfo{author}{Ma, X.}, \bibinfo{author}{Lu, G.},
  \bibinfo{year}{2017}.
\newblock \bibinfo{title}{Forecasting short-term subway passenger flow under
  special events scenarios using multiscale radial basis function networks}.
\newblock \bibinfo{journal}{Transportation Research Part C: Emerging
  Technologies} \bibinfo{volume}{77}, \bibinfo{pages}{306--328}.
\bibitem[{Liu et~al.(2020)Liu, Zheng, van Zuylen and Li}]{liu2020dynamic}
\bibinfo{author}{Liu, J.}, \bibinfo{author}{Zheng, F.}, \bibinfo{author}{van
  Zuylen, H.J.}, \bibinfo{author}{Li, J.}, \bibinfo{year}{2020}.
\newblock \bibinfo{title}{A dynamic od prediction approach for urban networks
  based on automatic number plate recognition data}.
\newblock \bibinfo{journal}{Transportation Research Procedia}
  \bibinfo{volume}{47}, \bibinfo{pages}{601--608}.
\bibitem[{Liu et~al.(2019a)Liu, Qiu, Li, Wang, Ouyang and
  Lin}]{liu2019contextualized}
\bibinfo{author}{Liu, L.}, \bibinfo{author}{Qiu, Z.}, \bibinfo{author}{Li, G.},
  \bibinfo{author}{Wang, Q.}, \bibinfo{author}{Ouyang, W.},
  \bibinfo{author}{Lin, L.}, \bibinfo{year}{2019}a.
\newblock \bibinfo{title}{Contextualized spatial--temporal network for taxi
  origin-destination demand prediction}.
\newblock \bibinfo{journal}{IEEE Transactions on Intelligent Transportation
  Systems} \bibinfo{volume}{20}, \bibinfo{pages}{3875--3887}.
\bibitem[{Liu et~al.(2019b)Liu, Liu and Jia}]{liu2019deeppf}
\bibinfo{author}{Liu, Y.}, \bibinfo{author}{Liu, Z.}, \bibinfo{author}{Jia,
  R.}, \bibinfo{year}{2019}b.
\newblock \bibinfo{title}{Deeppf: A deep learning based architecture for metro
  passenger flow prediction}.
\newblock \bibinfo{journal}{Transportation Research Part C: Emerging
  Technologies} \bibinfo{volume}{101}, \bibinfo{pages}{18--34}.
\bibitem[{Noursalehi et~al.(2021)Noursalehi, Koutsopoulos and
  Zhao}]{noursalehi2021dynamic}
\bibinfo{author}{Noursalehi, P.}, \bibinfo{author}{Koutsopoulos, H.N.},
  \bibinfo{author}{Zhao, J.}, \bibinfo{year}{2021}.
\newblock \bibinfo{title}{Dynamic origin-destination prediction in urban rail
  systems: A multi-resolution spatio-temporal deep learning approach}.
\newblock \bibinfo{journal}{IEEE Transactions on Intelligent Transportation
  Systems} .
\bibitem[{Proctor et~al.(2016)Proctor, Brunton and Kutz}]{proctor2016dynamic}
\bibinfo{author}{Proctor, J.L.}, \bibinfo{author}{Brunton, S.L.},
  \bibinfo{author}{Kutz, J.N.}, \bibinfo{year}{2016}.
\newblock \bibinfo{title}{Dynamic mode decomposition with control}.
\newblock \bibinfo{journal}{SIAM Journal on Applied Dynamical Systems}
  \bibinfo{volume}{15}, \bibinfo{pages}{142--161}.
\bibitem[{Ren and Xie(2017)}]{ren2017efficient}
\bibinfo{author}{Ren, J.}, \bibinfo{author}{Xie, Q.}, \bibinfo{year}{2017}.
\newblock \bibinfo{title}{Efficient od trip matrix prediction based on tensor
  decomposition}, in: \bibinfo{booktitle}{2017 18th IEEE International
  Conference on Mobile Data Management (MDM)}, \bibinfo{organization}{IEEE}.
  pp. \bibinfo{pages}{180--185}.
\bibitem[{Rowley et~al.(2009)Rowley, Mezić, Bagheri, Schlatter, Henningson
  et~al.}]{rowley2009spectral}
\bibinfo{author}{Rowley, C.W.}, \bibinfo{author}{Mezić, I.},
  \bibinfo{author}{Bagheri, S.}, \bibinfo{author}{Schlatter, P.},
  \bibinfo{author}{Henningson, D.}, et~al., \bibinfo{year}{2009}.
\newblock \bibinfo{title}{Spectral analysis of nonlinear flows}.
\newblock \bibinfo{journal}{Journal of fluid mechanics} \bibinfo{volume}{641},
  \bibinfo{pages}{115--127}.
\bibitem[{Scherl et~al.(2020)Scherl, Strom, Shang, Williams, Polagye and
  Brunton}]{scherl2020robust}
\bibinfo{author}{Scherl, I.}, \bibinfo{author}{Strom, B.},
  \bibinfo{author}{Shang, J.K.}, \bibinfo{author}{Williams, O.},
  \bibinfo{author}{Polagye, B.L.}, \bibinfo{author}{Brunton, S.L.},
  \bibinfo{year}{2020}.
\newblock \bibinfo{title}{Robust principal component analysis for modal
  decomposition of corrupt fluid flows}.
\newblock \bibinfo{journal}{Physical Review Fluids} \bibinfo{volume}{5},
  \bibinfo{pages}{054401}.
\bibitem[{Schmid(2010)}]{schmid2010dynamic}
\bibinfo{author}{Schmid, P.J.}, \bibinfo{year}{2010}.
\newblock \bibinfo{title}{Dynamic mode decomposition of numerical and
  experimental data}.
\newblock \bibinfo{journal}{Journal of fluid mechanics} \bibinfo{volume}{656},
  \bibinfo{pages}{5--28}.
\bibitem[{Shen et~al.(2020)Shen, Shao, Yu and Chen}]{shen2020hybrid}
\bibinfo{author}{Shen, L.}, \bibinfo{author}{Shao, Z.}, \bibinfo{author}{Yu,
  Y.}, \bibinfo{author}{Chen, X.}, \bibinfo{year}{2020}.
\newblock \bibinfo{title}{Hybrid approach combining modified gravity model and
  deep learning for short-term forecasting of metro transit passenger flows}.
\newblock \bibinfo{journal}{Transportation Research Record} ,
  \bibinfo{pages}{0361198120968823}.
\bibitem[{Shi et~al.(2015)Shi, Chen, Wang, Yeung, Wong and
  Woo}]{shi2015convolutional}
\bibinfo{author}{Shi, X.}, \bibinfo{author}{Chen, Z.}, \bibinfo{author}{Wang,
  H.}, \bibinfo{author}{Yeung, D.Y.}, \bibinfo{author}{Wong, W.K.},
  \bibinfo{author}{Woo, W.c.}, \bibinfo{year}{2015}.
\newblock \bibinfo{title}{Convolutional lstm network: A machine learning
  approach for precipitation nowcasting}.
\newblock \bibinfo{journal}{Advances in neural information processing systems}
  \bibinfo{volume}{28}, \bibinfo{pages}{802--810}.
\bibitem[{Sun et~al.(2015)Sun, Leng and Guan}]{sun2015novel}
\bibinfo{author}{Sun, Y.}, \bibinfo{author}{Leng, B.}, \bibinfo{author}{Guan,
  W.}, \bibinfo{year}{2015}.
\newblock \bibinfo{title}{A novel wavelet-svm short-time passenger flow
  prediction in beijing subway system}.
\newblock \bibinfo{journal}{Neurocomputing} \bibinfo{volume}{166},
  \bibinfo{pages}{109--121}.
\bibitem[{Toqu{\'e} et~al.(2016)Toqu{\'e}, C{\^o}me, El~Mahrsi and
  Oukhellou}]{toque2016forecasting}
\bibinfo{author}{Toqu{\'e}, F.}, \bibinfo{author}{C{\^o}me, E.},
  \bibinfo{author}{El~Mahrsi, M.K.}, \bibinfo{author}{Oukhellou, L.},
  \bibinfo{year}{2016}.
\newblock \bibinfo{title}{Forecasting dynamic public transport
  origin-destination matrices with long-short term memory recurrent neural
  networks}, in: \bibinfo{booktitle}{2016 IEEE 19th international conference on
  intelligent transportation systems (ITSC)}, \bibinfo{organization}{IEEE}. pp.
  \bibinfo{pages}{1071--1076}.
\bibitem[{Tu et~al.(2014)Tu, Rowley, Luchtenburg, Brunton and
  Kutz}]{tu2014dynamic}
\bibinfo{author}{Tu, J.H.}, \bibinfo{author}{Rowley, C.W.},
  \bibinfo{author}{Luchtenburg, D.M.}, \bibinfo{author}{Brunton, S.L.},
  \bibinfo{author}{Kutz, J.N.}, \bibinfo{year}{2014}.
\newblock \bibinfo{title}{On dynamic mode decomposition: Theory and
  applications}.
\newblock \bibinfo{journal}{Journal of Computational Dynamics}
  \bibinfo{volume}{1}, \bibinfo{pages}{391--421}.
\bibitem[{Wang et~al.(2020)Wang, Miao, Chen and Huang}]{wang2020multi}
\bibinfo{author}{Wang, S.}, \bibinfo{author}{Miao, H.}, \bibinfo{author}{Chen,
  H.}, \bibinfo{author}{Huang, Z.}, \bibinfo{year}{2020}.
\newblock \bibinfo{title}{Multi-task adversarial spatial-temporal networks for
  crowd flow prediction}, in: \bibinfo{booktitle}{Proceedings of the 29th ACM
  International Conference on Information \& Knowledge Management}, pp.
  \bibinfo{pages}{1555--1564}.
\bibitem[{Wang et~al.(2019a)Wang, Smola, Maddix, Gasthaus, Foster and
  Januschowski}]{wang2019deep}
\bibinfo{author}{Wang, Y.}, \bibinfo{author}{Smola, A.},
  \bibinfo{author}{Maddix, D.C.}, \bibinfo{author}{Gasthaus, J.},
  \bibinfo{author}{Foster, D.}, \bibinfo{author}{Januschowski, T.},
  \bibinfo{year}{2019}a.
\newblock \bibinfo{title}{Deep factors for forecasting}.
\newblock \bibinfo{journal}{arXiv preprint arXiv:1905.12417} .
\bibitem[{Wang et~al.(2019b)Wang, Yin, Chen, Wo, Xu and Zheng}]{wang2019origin}
\bibinfo{author}{Wang, Y.}, \bibinfo{author}{Yin, H.}, \bibinfo{author}{Chen,
  H.}, \bibinfo{author}{Wo, T.}, \bibinfo{author}{Xu, J.},
  \bibinfo{author}{Zheng, K.}, \bibinfo{year}{2019}b.
\newblock \bibinfo{title}{Origin-destination matrix prediction via graph
  convolution: a new perspective of passenger demand modeling}, in:
  \bibinfo{booktitle}{Proceedings of the 25th ACM SIGKDD International
  Conference on Knowledge Discovery \& Data Mining}, pp.
  \bibinfo{pages}{1227--1235}.
\bibitem[{Wei and Chen(2012)}]{wei2012forecasting}
\bibinfo{author}{Wei, Y.}, \bibinfo{author}{Chen, M.C.}, \bibinfo{year}{2012}.
\newblock \bibinfo{title}{Forecasting the short-term metro passenger flow with
  empirical mode decomposition and neural networks}.
\newblock \bibinfo{journal}{Transportation Research Part C: Emerging
  Technologies} \bibinfo{volume}{21}, \bibinfo{pages}{148--162}.
\bibitem[{Xiong et~al.(2020)Xiong, Ozbay, Jin and Feng}]{xiong2020dynamic}
\bibinfo{author}{Xiong, X.}, \bibinfo{author}{Ozbay, K.}, \bibinfo{author}{Jin,
  L.}, \bibinfo{author}{Feng, C.}, \bibinfo{year}{2020}.
\newblock \bibinfo{title}{Dynamic origin--destination matrix prediction with
  line graph neural networks and kalman filter}.
\newblock \bibinfo{journal}{Transportation Research Record}
  \bibinfo{volume}{2674}, \bibinfo{pages}{491--503}.
\bibitem[{Yu et~al.(2016)Yu, Rao and Dhillon}]{yu2016temporal}
\bibinfo{author}{Yu, H.F.}, \bibinfo{author}{Rao, N.},
  \bibinfo{author}{Dhillon, I.S.}, \bibinfo{year}{2016}.
\newblock \bibinfo{title}{Temporal regularized matrix factorization for
  high-dimensional time series prediction}, in: \bibinfo{booktitle}{Advances in
  neural information processing systems}, pp. \bibinfo{pages}{847--855}.
\bibitem[{Yu et~al.(2020)Yu, Zhang, Qian, Wang, Hu and Yin}]{yu2020low}
\bibinfo{author}{Yu, Y.}, \bibinfo{author}{Zhang, Y.}, \bibinfo{author}{Qian,
  S.}, \bibinfo{author}{Wang, S.}, \bibinfo{author}{Hu, Y.},
  \bibinfo{author}{Yin, B.}, \bibinfo{year}{2020}.
\newblock \bibinfo{title}{A low rank dynamic mode decomposition model for
  short-term traffic flow prediction}.
\newblock \bibinfo{journal}{IEEE Transactions on Intelligent Transportation
  Systems} .
\bibitem[{Zhang et~al.(2021a)Zhang, Xiao, Shen and Zhong}]{zhang2021dneat}
\bibinfo{author}{Zhang, D.}, \bibinfo{author}{Xiao, F.}, \bibinfo{author}{Shen,
  M.}, \bibinfo{author}{Zhong, S.}, \bibinfo{year}{2021}a.
\newblock \bibinfo{title}{{DNEAT}: {A} novel dynamic node-edge attention
  network for origin-destination demand prediction}.
\newblock \bibinfo{journal}{Transportation Research Part C: Emerging
  Technologies} \bibinfo{volume}{122}, \bibinfo{pages}{102851}.
\bibitem[{Zhang et~al.(2019a)Zhang, Rowley, Deem and
  Cattafesta}]{zhang2019online}
\bibinfo{author}{Zhang, H.}, \bibinfo{author}{Rowley, C.W.},
  \bibinfo{author}{Deem, E.A.}, \bibinfo{author}{Cattafesta, L.N.},
  \bibinfo{year}{2019}a.
\newblock \bibinfo{title}{Online dynamic mode decomposition for time-varying
  systems}.
\newblock \bibinfo{journal}{SIAM Journal on Applied Dynamical Systems}
  \bibinfo{volume}{18}, \bibinfo{pages}{1586--1609}.
\bibitem[{Zhang et~al.(2021b)Zhang, Che, Chen, Ma and He}]{zhang2021short}
\bibinfo{author}{Zhang, J.}, \bibinfo{author}{Che, H.}, \bibinfo{author}{Chen,
  F.}, \bibinfo{author}{Ma, W.}, \bibinfo{author}{He, Z.},
  \bibinfo{year}{2021}b.
\newblock \bibinfo{title}{Short-term origin-destination demand prediction in
  urban rail transit systems: A channel-wise attentive split-convolutional
  neural network method}.
\newblock \bibinfo{journal}{Transportation Research Part C: Emerging
  Technologies} \bibinfo{volume}{124}, \bibinfo{pages}{102928}.
\bibitem[{Zhang et~al.(2020)Zhang, Chen, Cui, Guo and Zhu}]{zhang2020deep}
\bibinfo{author}{Zhang, J.}, \bibinfo{author}{Chen, F.}, \bibinfo{author}{Cui,
  Z.}, \bibinfo{author}{Guo, Y.}, \bibinfo{author}{Zhu, Y.},
  \bibinfo{year}{2020}.
\newblock \bibinfo{title}{Deep learning architecture for short-term passenger
  flow forecasting in urban rail transit}.
\newblock \bibinfo{journal}{IEEE Transactions on Intelligent Transportation
  Systems} .
\bibitem[{Zhang et~al.(2019b)Zhang, Chen, Wang and Liu}]{zhang2019short}
\bibinfo{author}{Zhang, J.}, \bibinfo{author}{Chen, F.}, \bibinfo{author}{Wang,
  Z.}, \bibinfo{author}{Liu, H.}, \bibinfo{year}{2019}b.
\newblock \bibinfo{title}{Short-term origin-destination forecasting in urban
  rail transit based on attraction degree}.
\newblock \bibinfo{journal}{IEEE Access} \bibinfo{volume}{7},
  \bibinfo{pages}{133452--133462}.

\end{thebibliography}
\end{document}